\documentclass[12pt]{article}
\usepackage{amsfonts}
\usepackage{amsmath}
\usepackage{amssymb}
\usepackage{graphicx}
\newcommand{\imineq}[2]{\vcenter{\hbox{\includegraphics[height=#2ex]{#1}}}}
\usepackage{cite}
\usepackage{tensor}
\usepackage{cases}
\usepackage{multirow}
\usepackage{booktabs}
\usepackage{color}
\definecolor{pink}{RGB}{233, 0, 100}
\usepackage{ulem}
\usepackage{appendix}
\usepackage[breaklinks=true,colorlinks=true
,urlcolor=blue
,anchorcolor=blue
,citecolor=blue
,filecolor=blue
,linkcolor=blue
,menucolor=blue
,pagecolor=blue
,linktocpage=true
,pdfproducer=medialab
,pdfa=true
]{hyperref}
\usepackage[utf8]{inputenc}
\usepackage[margin=2cm]{geometry}
\usepackage{geometry}

\usepackage{enumitem}
\def \bea{\begin{eqnarray}}
\def \eea{\end{eqnarray}}
\def \be{\begin{equation}}
\def \ee{\end{equation}}
\def \nn {\nonumber}

\def \tr{\text{tr}}
\def \T{\mathcal{T}}
\def \mo{\mathcal{O}}
\def \e{\epsilon}
\usepackage{mathrsfs}
\usepackage{amsthm}
\newtheoremstyle{thry}
{6pt}
{6pt}
{\itshape}
{}
{\bfseries}
{:}
{.5em}
{}
\newtheorem{proposition}{Proposition}
\newtheorem{theorem}{Theorem}
\newtheorem{corollary}{Corollary}
\newtheoremstyle{remar}
{6pt}
{6pt}
{\upshape}
{}
{\bfseries}
{}
{.5em}
{}
\theoremstyle{remar}
\newtheorem*{remark}{Remark:}
\makeatletter
\renewenvironment{proof}[1][\proofname]{\par
\pushQED{\qed}%
\normalfont \topsep6\p@\@plus6\p@\relax
\trivlist
\item\relax
{\itshape
#1\@addpunct{.}}\hspace\labelsep\ignorespaces
}{%
\popQED\endtrivlist\@endpefalse
}
\makeatother
\usepackage[font=small]{subfig}
\usepackage[skip=8pt,labelfont={bf,sc},textfont=small]{caption}
\begin{document}
\begin{titlepage}
\vspace{0.5cm}
\begin{center}
{\Large \bf {Constructible reality condition of  pseudo entropy via pseudo-Hermiticity}}
\lineskip .75em
\vskip 2.5cm
{\large Wu-zhong Guo$^{a,}$\footnote{wuzhong@hust.edu.cn}, Song He$^{b,c,}$\footnote{hesong@jlu.edu.cn; corresponding author},  Yu-Xuan Zhang$^{b,}$\footnote{yuxuanz20@mails.jlu.edu.cn; corresponding author}}
\vskip 2.5em
{\normalsize\it $^{a}$School of Physics, Huazhong University of Science and Technology\\
Luoyu Road 1037, Wuhan, Hubei 430074, China\\
$^{b}$Center for Theoretical Physics and College of Physics, Jilin University,\\ Changchun 130012, People's Republic of China\\
$^{c}$Max Planck Institute for Gravitational Physics (Albert Einstein Institute),\\
Am M\"uhlenberg 1, 14476 Golm, Germany
}
\vskip 2.0em
\textcolor{blue}{\textit{\large Dedicated to Hermann Nicolai on the occasion of his $70^{th}$  birthday.}}
\end{center}
\begin{abstract}
As a generalization of entanglement entropy, pseudo entropy is not always real. The real-valued pseudo entropy has promising applications in holography and quantum phase transition. We apply the notion of pseudo-Hermiticity to formulate the reality condition of pseudo entropy. We find the general form of the transition matrix for which the eigenvalues of the reduced transition matrix possess real or complex pairs of eigenvalues. Further, we find a class of transition matrices for which the pseudo (R\'enyi) entropies are non-negative. Some known examples which give real pseudo entropy in quantum field theories can be explained in our framework. Our results offer a novel method to generate the transition matrix with real pseudo entropy. Finally, we show the reality condition for pseudo entropy is related to the Tomita-Takesaki modular theory for quantum field theory.
\end{abstract}
\end{titlepage}

\baselineskip=0.7cm

\tableofcontents
\section{Introduction}

Entanglement entropy (EE), as an entanglement measure, has been investigated in many aspects \cite{Vidal:2002rm,Kitaev:2005dm,Levin:2006zz,Srednicki:1993im,Eisert:2008ur,Calabrese:2004eu,Ryu:2006bv,Hubeny:2007xt}. Specially, in the context of AdS/CFT \cite{Maldacena:1997re,Gubser:1998bc,Witten:1998qj}, entanglement plays an important role in understanding the emergence of geometry \cite{VanRaamsdonk:2010pw,Rangamani:2016dms}, subregion/subregion duality \cite{Almheiri:2014lwa,Dong:2016eik} and information paradox of black hole \cite{Penington:2019npb,Almheiri:2019psf}.

Recently,  a generalization of EE, called \textit{pseudo entropy}, was introduced in \cite{Nakata:2020luh} via AdS/CFT and post-selection, which may bring us a new understanding of the role of entanglement in quantum field theory (QFT) or gravity.
Given a system whose Hilbert space can be divided as $H=H_{A}\otimes H_{\bar{A}}$, the pseudo entropy of subsystem $A$ is defined as
\bea
S(\mathcal{T}_A^{\psi|\phi})=-\tr[\mathcal{T}_A^{\psi|\phi}\log\mathcal{T}_A^{\psi|\phi}],\label{eq1pseudo}
\eea
where $\mathcal{T}_A^{\psi|\phi}=\tr_{\bar A}\left[\mathcal{T}^{\psi|\phi}\right]$, called \textit{reduced transition matrix}, is defined from the transition matrix consisting of two nonorthogonal pure states $|\psi\rangle$ and $|\phi\rangle$
\bea
\mathcal{T}^{\psi|\phi}=\frac{|\psi\rangle\langle\phi|}{\langle\phi|\psi\rangle}.\label{eq2:transitionMa}
\eea
Similar to the EE, in practice,  especially in QFTs, one uses the replica trick \cite{Calabrese:2004eu} and computes the so-called \textit{pseudo R\'enyi entropy} instead
\bea\label{pseudoRenyi}
S^{(n)}(\mathcal{T}_A^{\psi|\phi})=\frac{1}{1-n} \log \tr[ (\mathcal{T}_A^{\psi|\phi})^n].
\eea
The pseudo entropy \eqref{eq1pseudo} can be obtained by taking the limit of $n\to1$ for the above expression.

There are manifold interests driving the study of this quantity. See  \cite{Mollabashi:2020yie,Camilo:2021dtt,Mollabashi:2021xsd,Nishioka:2021cxe,Goto:2021kln,Miyaji:2021lcq,Akal:2021dqt,Mukherjee:2022jac,Guo:2022sfl,Ishiyama:2022odv,Bhattacharya:2022wlp,Doi:2022iyj} for recent studies. The first comes from holography. In QFTs, the pseudo R\'enyi entropy can be formulated in the language of the path integral. According to the AdS/CFT, the path integral in CFTs can be translated to the gravitational path integral in AdS. Based on this, it is proposed in \cite{Nakata:2020luh} that the pseudo entropy of a boundary subsystem is dual to the area of a minimal surface in Euclidean asymptotically
 AdS spacetime. It is one of the main motivations of \cite{Nakata:2020luh} to propose this novel quantity.   The second interest comes from the experiments of post-selection in quantum measurements \cite{PhysRev.134.B1410,Aharonov:1988xu,RevModPhys.86.307}.
It is argued and proved in a two-qubit model that the pseudo entropy characterizes the averaged number of  EPR pairs that could have been distilled in the pre-and post-selected systems \cite{Nakata:2020luh}. Since the
post-selection measurement can be incomplete \cite{aharonov1991complete,aharonov2008two}, motivated by this, in this article, we would also like to consider a generalization of the transition matrices \eqref{eq2:transitionMa} to mixed states. \textcolor{black}{Given two general density matrices $\rho_1$ and $\rho_2$ $(\tr[\rho_1\rho_2]\neq0)$ in Hilbert space, we construct the following operator as a mixed-state generalization of the transition matrix.}
\begin{align}
\textcolor{black}{X:=\frac{\rho_1\rho_2}{\tr[\rho_1\rho_2]}}.\label{defofX}
\end{align}
 \textcolor{black}{Eq.\eqref{defofX} can be understood as the system being pre-selected and post-selected to mixed states $\rho_1$ and $\rho_2$  instead of two pure states, respectively. Note that Eq.\eqref{defofX} is reduced to Eq.\eqref{eq2:transitionMa} when $\rho_{1,2}$ are pure.} The third interest in studying pseudo entropy  comes from quantum many-body systems.  The study of pseudo entropy in Ising spin chain \cite{Mollabashi:2020yie} indicates that pseudo entropy can be taken as a new order parameter in quantum many-body systems, just like EE \cite{Kitaev:2005dm,Levin:2006zz}.

All the above alluring physical interpretations and application of pseudo entropy are implicitly based on such a fact: The pseudo entropy $S(\mathcal{T}_A^{\psi|\phi})$ of a subsystem is real. However, $S(\mathcal{T}_A^{\psi|\phi})$ should be complex-valued for generic choices of the initial and finite states since the reduced transition matrix, in general, is non-Hermitian. When we construct the initial and final state by a Euclidean path integral with a real-valued action, it seems reasonable that $S(\mathcal{T}_A^{\psi|\phi})$ can be positive. Whereas, for more generic cases, such as the insertion of non-Hermitian operators in Euclidean path integrals, QFTs with Lorentzian signature, and the actual finite-dimensional quantum systems, we have to ask what kind of transition matrices generate positive or non-negative pseudo entropy. It is the central motivation for the present article.

The problem is closely related to non-Hermitian physics, which has been extensively studied recently; see the review \cite{Ashida:2020dkc,bender2019pt}. It is found the eigenvalues of the non-Hermitian Hamiltonian can be real. The system satisfying parity-time ($\mathcal{PT}$) symmetry is one of the most important classes \cite{Bender:1998ke,Bender:2007nj}. The notion of pseudo-Hermiticity is handy to characterize a class of non-Hermitian matrices with real eigenvalues \cite{Mostafazadeh:2001jk,Mostafazadeh:2001nr,Mostafazadeh:2008pw} if the matrix has a complete biorthonormal eigenbasis. An operator $M$ is said to be $\eta$-pseudo-Hermitian if there exists a Hermitian invertible operator $\eta$ such that
\bea
M^\dagger=\eta M \eta^{-1}.\label{phc}
\eea
If $\eta$ is the identity, the pseudo-Hermitian condition reduces to Hermiticity.
The necessary and sufficient conditions for the pseudo-Hermiticity of $M$  are given by the following theorem \cite{Mostafazadeh:2001jk}.
\begin{theorem}\label{theor1} An operator $M$ with a complete biorthonormal eigenbasis and a discrete spectrum is pseudo-Hermitian if and only if one of the following conditions hold:
\begin{enumerate}
\item The eigenvalues of $M$ are real.
\item The complex eigenvalues come in complex conjugate pairs, and the degeneracy of the eigenvalues are same.
\end{enumerate}
\end{theorem}
\begin{remark}
{The existence of a biorthonormal eigenbasis in Hilbert space with respect to $M$ is equivalent to  $M$ being diagonalizable. According to the Theorem \ref{theor1}, when the reduced transition matrix $X_A\equiv\tr_{\bar A}[X]$ is diagonalizable, the necessary and sufficient condition of $\tr[(X_A)^n]$ being real is that $X_A$ is $\eta_A$-pseudo-Hermitian \footnote{When $X_A$ is non-diagonalizable,  a pseudo-Hermitian $X_A$ can still give a real $\tr[(X_A)^n]$, but the converse is not necessarily true. We show an example in Appendix \ref{feiduijiao} where $\tr[(X_A)^n]$ is real, but $X_A$ is not pseudo-Hermitian.}. 
The pseudo R\'enyi entropy $S^{(n)}(X_A)$, however, may not be real in this case. Apparently, to guarantee the reality of $S^{(n)}(X_A)$, one should require $\tr[(X_A)^n]>0$, which gives more constraints on $\eta_A$.

Many studies on $\mathcal{P}\mathcal{T}$-symmetric or pseudo-Hermitian Hamiltonian mainly focus on finite-dimensional systems. It is expected the main results can also be applicable to QFTs, see, e.g., \cite{PhysRevD.71.025014,PhysRevD.98.125003,Ashida_2017,Bender:2021fxa} and references therein.}
\end{remark}
\section{General form of transition matrix }
\textcolor{black}{We seek to construct transition matrix $X$ $\eqref{defofX}$ with positive $\tr[(X_A)^n]$ by means of pseudo-Hermiticity. We start with a  proposition for transition matrices that strongly suggests a potential connection between pseudo entropy and pseudo-Hermiticity.}
\begin{proposition}\label{propoforph}
  \textcolor{black}{All transition matrices \eqref{defofX} in finite-dimensional Hilbert space are pseudo-Hermitian.}
\end{proposition}

\begin{proof}
  \textcolor{black}{The statement comes from the fact that the product of any two positive semi-definite matrices can be diagonalized and possesses non-negative eigenvalues. One can construct a complete biorthonormal eigenbasis within the finite-dimensional Hilbert space by utilizing the eigenvectors of a given diagonalizable matrix. Consequently, all finite-dimensional transition matrices of the form \eqref{defofX} have a complete biorthonormal eigenbasis and non-negative eigenvalues,  satisfying the conditions required for Theorem \ref{theor1}.  Then, it follows that all transition matrices are pseudo-Hermitian.}
\end{proof}
\begin{remark}
  \textcolor{black}{Building to the above proposition, for any transition matrix of the form \eqref{defofX} in finite-dimensional Hilbert space, we can construct a $\eta$ matrix that satisfies Eq.\eqref{phc}  by using its biorthonormal eigenbasis \cite{Mostafazadeh:2001jk}. As a special case, for a pure state transition matrix \eqref{eq2:transitionMa} in $d$-dimensional Hilbert space, a valid $\eta$ matrix satisfying  Eq.\eqref{phc} can be found as}\footnote{\textcolor{black}{The $|\psi\rangle,~|\psi_1\rangle,~|\psi_2\rangle,...,|\psi_{d-1}\rangle$ in \eqref{etapt} form an orthonormal basis in the Hilbert space. Additionally, $|\psi\rangle$ and $|\psi_1\rangle$ satisfy the equation $a |\psi\rangle + b |\psi_1\rangle = |\phi\rangle$.}}
  \begin{align}
  \textcolor{black}{\eta=\frac{|\phi\rangle\langle\phi|}{|\langle\phi|\psi\rangle|^2}+\sum_{i=1}^{d-1}|\psi_i\rangle\langle\psi_i|,}\label{etapt}
  \end{align}
\textcolor{black}{which has an inverse}
\begin{align}
\textcolor{black}{\eta^{-1}=|\psi\rangle\langle\psi|+\Big(|\psi_1\rangle-\frac{b^*}{a^*}|\psi\rangle\Big)\Big(\langle\psi_1|-\frac{b}{a}\langle\psi|\Big)+\sum_{i=2}^{d-1}|\psi_i\rangle\langle\psi_i|.}
\end{align}
\end{remark}

\textcolor{black}{Although all transition matrices are pseudo-Hermitian, it does not guarantee  that the reduced transition matrices share the same property. The two propositions below jointly  establish a necessary and sufficient condition for a reduced transition matrix to be pseudo-Hermitian.}
\begin{proposition}\label{propo1}Any operator $\mo$ can be decomposed as
\bea
\mo=\mo_1+i \mo_2,
\eea
where $\mo_1$ and $\mo_2$ are $\eta$-pseudo-Hermitian operators, $\eta$ can be any Hermitian invertible operator.
\end{proposition}

\begin{proof} For any operator $\mo$, we can divide it into two parts
\bea\label{decom}
\mo=\frac{\mo+\eta^{-1}\mo^{\dagger}\eta}{2}+i \frac{\mo-\eta^{-1}\mo^{\dagger}\eta}{2i},
\eea
where $\eta$ is any Hermitian invertible operator. For $\eta$ being identity, (\ref{decom}) is the well-known result that any operator can be decomposed as linear combinations of two Hermitian operators.  Where $\mo_1= \frac{\mo+\eta^{-1}\mo^{\dagger}\eta}{2}$ and $\mo_2= \frac{\mo-\eta^{-1}\mo^{\dagger}\eta}{2i}$.
\end{proof}

\begin{proposition}\label{propo2}  $X_{A(\bar A )}$ is $\eta_{A(\bar A)}$-pseudo-Hermitian, if and only if the transition matrix $X$ can be written as
\bea\label{transitionXcondition}
X=X_1+i X_2,
\eea
where $X_1$ and $X_2$ are both $\eta$-pseudo-Hermitian with $\eta=\eta_A \otimes \eta_{\bar A}$. Further, $X_2$ satisfies $\tr_{\bar A(A)}X_2=0$.
\end{proposition}

\begin{proof} Using the result of Proposition \ref{propo1}, let's define the operator
\bea
&&X_1:=\frac{1}{2}(X+\eta^{-1} X^\dagger \eta ),\nonumber \\
&&X_2:= \frac{i}{2}(\eta^{-1} X^\dagger \eta - X).
\eea
Since $X_{A(\bar A )}$ is $\eta_{A(\bar A)}$-pseudo-Hermitian, we have
\bea
\tr_{\bar A} X_2=\frac{i}{2}[\eta_A^{-1}(\tr _{\bar A}X^\dagger)\eta_{A}-\tr_{\bar A}X]=\frac{i}{2}(\eta_A^{-1}X_{A}^\dagger \eta_{A}-X_A)=0.
\eea
Similarly, one could show $\tr_A X_2=0$.

If $X$ can be written as (\ref{transitionXcondition}) and $\tr_{\bar A} X_2=0$, we have $X_{A}=\tr_{\bar A} X_1$ and
\bea
X_{A}^\dagger=\tr_{\bar A}X^\dagger=\tr_{\bar A} (\eta X_1 \eta^{-1})=\eta_A (\tr_{\bar A} X_1) \eta_A^{-1}=\eta_A X_A \eta_A^{-1}.
\eea
Thus $X_A$ is $\eta_A$-pseudo-Hermitian.
 Similarly, we can show $X_{\bar A}$ is $\eta_{\bar A}$-pseudo-Hermitian.
\end{proof}

An obvious corollary of Proposition $\ref{propo2}$ is that a $\eta$-pseudo-Hermitian transition matrix with $\eta=\eta_A\otimes\eta_{\bar A}$ generates  pseudo-Hermitian reduced transition matrices. Let's first focus on the pure state transition matrix $\mathcal{T}^{\psi|\phi}$ and show how to construct a $\eta$-pseudo-Hermitian transition matrices. The $\eta$-pseudo-Hermiticity gives constraints on the pure states $|\psi\rangle$ and $|\phi\rangle$.
We have the following theorem.

\begin{theorem}\label{theor2}
A transition matrix $\mathcal{T}^{\psi|\phi}$ is $\eta$-pseudo-Hermitian , if and only if it can be written as follows.
\bea\label{Tpseudo}
\mathcal{T}^{\psi|\phi}=\frac{|\psi\rangle \langle \psi| \eta}{\langle \psi|\eta |\psi\rangle},
\eea
\textcolor{black}{where $\eta$ is both Hermitian and invertible.}
 \end{theorem}

\begin{proof}
Assume the transition matrix $\mathcal{T}^{\psi|\phi}=\frac{|\psi\rangle\langle\phi|}{\langle\phi|\psi\rangle}$ is $\eta$-pseudo-Hermitian, we have
\bea
\frac{\eta |\psi\rangle \langle \phi| \eta^{-1}}{\langle \phi| \psi\rangle}=\frac{|\phi\rangle \langle \psi|}{\langle \psi |\phi\rangle}.
\eea
It leads to
\bea
\eta |\psi\rangle =\frac{\langle \psi| \eta |\psi\rangle}{\langle\psi  |\phi\rangle} |\phi\rangle.
\eea
Taking the above formula into (\ref{eq2:transitionMa}), one can see the transition matrix $\mathcal{T}^{\psi|\phi}$ is of the form \eqref{Tpseudo}. \textcolor{black}{On the other hand, it is easy to show that if a transition matrix $\mathcal{T}$ takes the form of Eq.\eqref{Tpseudo}, where $\eta$ is an invertible Hermitian matrix, $\mathcal{T}$ is $\eta$-pseudo-Hermitian.}
\end{proof}

\begin{remark}
In the Theorem \ref{theor2}, the pure states $|\psi\rangle$ and $|\phi\rangle$ of the transition matrix $\mathcal{T}^{\psi|\phi}$ are assumed to be non-orthogonal. For the case of $\langle \phi |\psi\rangle =0$, one may consider the matrix $\mathcal{T}^{'\psi|\phi}=|\psi\rangle \langle \phi|$. If it is $\eta$-pseudo-Hermitian, we have
\bea
\eta|\psi\rangle \langle \phi|\eta^{-1}= |\phi\rangle \langle \psi|.
\eea
To satisfy the above condition, it is necessary that $\langle \psi |\eta|\phi\rangle=\langle \phi |\eta|\psi\rangle\ne 0$, otherwise we would have $\eta |\psi\rangle =0$, which is impossible for $|\psi\rangle \ne 0$. Further, we have the relation
\bea
|\phi\rangle =\frac{\eta|\psi\rangle}{\langle \psi|\eta |\phi\rangle}.
\eea
Therefore, the transition matrix $|\psi\rangle \langle \phi|$ is given by
\bea\label{TpseudoAlt}
\mathcal{T}^{'\psi|\phi}=\frac{|\psi\rangle \langle \psi| \eta}{\langle \psi|\eta |\phi\rangle},
\eea
for the case $\langle \psi|\phi\rangle=0$.
\end{remark}
A $\eta$-pseudo-Hermitian transition matrix $\mathcal{T}^{\psi|\phi}$ with $\eta=\eta_{A}\otimes\eta_{\bar A}$  generates  pseudo-Hermitian reduced transition matrices $\mathcal{T}^{\psi|\phi}_{A}$ and $\mathcal{T}^{\psi|\phi}_{\bar A}$. Whereas, the pseudo-Hermiticity of  $\mathcal{T}^{\psi|\phi}_{A(\bar A)}$ only guarantees $\tr[(\mathcal{T}^{\psi|\phi}_{A(\bar A)})^n]$ is real, not necessarily positive. One can expect to give more constraints on $\eta_{A(\bar A)}$ to make $\tr[(\mathcal{T}^{\psi|\phi}_{A(\bar A)})^n]>0$. The following theorem provides a construction in which the eigenvalues of the reduced transition matrix are all non-negative and not all zeros so that $\tr[(\mathcal{T}^{\psi|\phi}_{A(\bar A)})^n]>0$.

\begin{theorem}\label{corol3}
$\mathcal{T}^{\psi|\phi}$ is $\eta$-pseudo-Hermitian with $\eta=\eta_A\otimes \eta_{\bar A}$.
\begin{enumerate}
    \item  If $\eta_{A(\bar A)}$ is positive or negative definite operator, the eigenvalues of $\mathcal{T}^{\psi|\phi}_{A(\bar A)}$ are real.
    \item  {If  $\eta_{A}$ is positive or negative definite  and $\eta_{\bar{A}}$ is positive or negative definite too, then the eigenvalues of $\mathcal{T}^{\psi|\phi}_{A(\bar A)}$ are non-negative and not all zeros.}
\end{enumerate}
\end{theorem}

\begin{proof}
In general, $\mathcal{T}^{\psi|\phi}_{A(\bar A)}$ are expected to have complex eigenvalues  for arbitrary Hermitian $\eta_{A(\bar A)}$. Let us define
\bea
\tilde{\mathcal{T}}_A:= \frac{\tr_{\bar A}\left[|\psi\rangle \langle \psi|\eta_{\bar A}\right]}{\langle \psi| \eta |\psi\rangle}.
\eea
It is obvious that $\tilde{\mathcal{T}}_A$ is a Hermitian operator.
By using (\ref{Tpseudo}) we have
\bea
\mathcal{T}^{\psi|\phi}_A=\tr_{\bar A}\mathcal{T}^{\psi|\phi} =\tilde{\mathcal{T}}_A  \eta_A.
\eea
Assume that $\eta_{A}$ is a positive (negative) definite operator. There exists a Hermitian (skew-Hermitian) invertible operator $\eta_{A}^{1/2}$ such that  $(\eta_{ A}^{1/2})^2=\eta_{ A}$. 
Thus we have
\bea\label{co3}
\eta_A^{1/2} \mathcal{T}^{\psi|\phi}_A \eta_A^{-1/2}=\eta_{A}^{1/2} \tilde{\mathcal{T}}_A \eta_A^{1/2}.
\eea
$\mathcal{T}^{\psi|\phi}_A$ is similar to the operator on the right-hand side of the above equation, which is  Hermitian. Thus, the eigenvalues of $\mathcal{T}^{\psi|\phi}_A$ are real.

{If further assuming $\eta_{\bar A}$ is positive (negative) definite, we have}
\bea
\tilde{\mathcal{T}}_A=\frac{\tr_{\bar A}\left[\eta_{\bar A}^{1/2} |\psi\rangle \langle \psi| \eta_{\bar A}^{1/2}\right]}{\langle \psi| \eta |\psi\rangle},\label{ttaat}
\eea
where we define the  Hermitian (skew-Hermitian) invertible operator $\eta_{\bar A}^{1/2}$ and use the cyclic property of partial trace. It is not hard to show $\eta_A^{1/2}\tilde{\mathcal{T}}_A\eta_A^{1/2}$ is always positive semi-definite and not null, the eigenvalues of which are non-negative and not all zeros.\footnote{\textcolor{black}{Please refer to  Appendix \ref{proofoftheo3} for a proof.}} Therefore, using (\ref{co3}), we have proved the eigenvalues of $\mathcal{T}^{\psi|\phi}_A$ are non-negative and not all zeros. By a similar process, one could show the eigenvalues of $\mathcal{T}^{\psi|\phi}_{\bar A}$ are all non-negative and not all zeros.
\end{proof}
\textcolor{black}{Since the pure state transition matrix \eqref{eq2:transitionMa} has unit trace, for a transition matrix $\mathcal{T}^{\psi|\phi}$ satisfying condition 2 in Theorem \ref{corol3}, we have $0<\tr[(\mathcal{T}_A^{\psi|\phi})^n]\leq1~(n\geq2)$. Hence, we construct a class of transition matrices  for which the pseudo (R\'enyi) entropies \eqref{pseudoRenyi} are non-negative.}
\begin{remark}
    The above theorem can be generalized to the case where $|\psi\rangle$ and $|\phi\rangle$ are orthogonal to each other. Building on \eqref{TpseudoAlt}, for the second case ($\eta_{A}$ is positive or negative definite  and $\eta_{\bar{A}}$ is positive or negative definite too), one has an expression similar to Eq.\eqref{co3},
    \begin{align}
    \eta_A^{1/2}\mathcal{T}_{A}^{'\psi|\phi}\eta_A^{-1/2}=\eta_A^{1/2}\tilde{T}'_A\eta_A^{1/2},\label{t'21}
    \end{align}
    where $\tilde{T}'_A=\tr_{\bar A}\left(\eta_{\bar A}^{1/2}|\psi\rangle \langle \psi|\eta_{\bar A}^{1/2}\right)/\langle \psi| \eta |\phi\rangle$ is a positive or negative semi-definite matrix depending on $\eta_{\bar A}$ and the value of $\langle \psi| \eta |\phi\rangle$. Thus the right-hand side of \eqref{t'21} is positive or negative semi-definite, depending on $\eta_{A(\bar A)}$ and $\langle \psi| \eta |\phi\rangle$. In what follows, we assume that it is positive semi-definite, which can always be achieved by redefining $|\phi\rangle$ as $-|\phi\rangle$.
\end{remark}

%
%
So far, we have studied the construction of the pure state transition matrix using the pseudo-Hermiticity to obtain non-negative pseudo (R\'enyi) entropy.  The following theorem shows that, for generic cases (i.e., the transition matrix takes the form \eqref{defofX}), the construction we find is closely related to the results of the pure state transition matrix.
\begin{theorem}\label{corol2}
Any diagonalizable $\eta$-pseudo-Hermitian matrix $M$ can be expressed as
\bea\label{coro2}
M=\sum_i m_i  \mathcal{T}^{\psi_i| \phi_i},
\eea
where $m_i$ are real, $ \mathcal{T}^{\psi_i| \phi_i}$ are $\eta$-pseudo-Hermitian transition matrices between two pure states $|\psi_i\rangle $ and $|\phi_i\rangle$ , which take the form (\ref{Tpseudo}) or (\ref{TpseudoAlt}).
\end{theorem}
The proof can be found in Appendix \ref{corollary2}. \textcolor{black}{By combining Proposition \ref{propoforph} with the above theorem, we can  conclude that any mixed state transition matrix $X$  of the form $\eqref{defofX}$ in finite-dimensional Hilbert space can be decomposed into the sum of a series of pure state transition matrices of the form (\ref{Tpseudo}) or (\ref{TpseudoAlt}), that is,}
\begin{align}
\textcolor{black}{X=\frac{\rho_1\rho_2}{\tr[\rho_1\rho_2]}=\sum_i m_i  \mathcal{T}^{\psi_i| \phi_i},\label{recast1}}
\end{align}
\textcolor{black}{where $m_i$ are real and both $X$ and $ \mathcal{T}^{\psi_i| \phi_i}$ are pseudo-Hermitian under some $\eta$ matrix. We may further recast $X$ into a more general form than $\eqref{recast1}$ by combining Theorem \ref{corol2} with Proposition \ref{propo1}.}
\begin{corollary}\label{coro1111x}
 For any transition matrix $X$ of form $\eqref{defofX}$, assuming that there is a Hermitian invertible matrix $\eta$ such that $X\pm\eta^{-1}X^\dagger \eta$ are diagonalizable, then $X$ can be expressed as
\bea\label{theorem2}
X=\sum_i x^1_i \mathcal{T}_1^{\psi_i|\phi_i}+i \sum_j x^2_j  \mathcal{T}_2^{\psi_j|\phi_j},
\eea
where $\mathcal{T}_{\alpha}^{\psi_i|\phi_i}$  are $\eta$-pseudo-Hermitian matrices which take the form (\ref{Tpseudo}) or (\ref{TpseudoAlt}) and $x^\alpha_i$ are real ($\alpha=1,2$).
\end{corollary}
\textcolor{black}{Combining Proposition \ref{propoforph}, Theorem \ref{corol3}, and Theorem \ref{corol2}, we get the second corollary as a summary of our results for the mixed state transition matrix.}
\begin{corollary}\label{corol4}
\textcolor{black}{The transition matrix $X$ taking the form \eqref{defofX} is $\eta$-pseudo-Hermitian with $\eta=\eta_A\otimes \eta_{\bar A}$. Assume both $\eta_{A}$ and $\eta_{\bar A}$ are positive or negative definite operators. If $m_i$ in Eq.\eqref{recast1} $>0$, the eigenvalues of $X_A$ and $X_{\bar A}$ are non-negative and not all zeros.}
\end{corollary}
\begin{remark}
One can show Corollary \ref{corol4} by using the fact that the linear combinations of positive semi-definite operators with positive coefficients are still positive semi-definite. It can be slightly generalized. Once some coefficients, namely the elements in a subset $\{m_a\}$, are negative, and the transition matrix $\mathcal{T}^{\psi_a|\phi_a}$ satisfies $\tr_{A(\bar A)} \sum_a m_a\mathcal{T}^{\psi_a|\phi_a}=0$, one can show the eigenvalues of $X_{A(\bar A)}$ are non-negative and not all zeros.

Since $X$ has unit trace,  one can show $0<\tr[(X_{A(\bar{A})})^n]\leq1$. Then Corollary \ref{corol4} provides an approach to generate transition matrix $X$, such that the pseudo R\'enyi entropies of $X_{A(\bar A)}$ are non-negative. The above condition is sufficient to have non-negative pseudo R\'enyi entropy. In Appendix \ref{nonpositveexample}, we show an example of a finite dimension.  $\eta_{A(\bar A)}$ are neither positive nor negative, the pseudo R\'enyi entropy is positive for integers $n\ge 2$.
\end{remark}
\section{Construction of the transition matrix}
Next, we shall explore the concrete implementation of pseudo-Hermitian transition matrices satisfying Theorem \ref{corol3} in finite-dimensional quantum systems and QFTs, respectively.\footnote{\textcolor{black}{Considering a lattice-regularized QFT may ensure that Theorem \ref{corol3} can be applied legally.}} Before that, we would like to review the algebraic view of QFTs, which will be useful in the following discussion.

 Let's consider a bipartite system whose Hilbert space $\mathcal{H}$ is divided into two subspaces $\mathcal{H}_A$ and $\mathcal{H}_{\bar A}$ of the same dimension. Let's denote $\mathcal{R}_{A(\bar A)}$ to be the algebra of operators working on $\mathcal{H}_{A(\bar A)}$. The   algebra of operators for the total system is given by $\mathcal{R}=\mathcal{R}_A\otimes \mathcal{R}_{\bar A}$. Choosing a reference state $|\Psi\rangle$, one could generate the states in $\mathcal{H}$ by acting the operators in $\mathcal{R}_{A(\bar A)}$ on $|\Psi\rangle$. A state $|\Psi\rangle$ in $\mathcal{H}$ is said to be cyclic for an algebra  such as $\mathcal{R}_A$ if the set $\{ a|\Psi\rangle\}$, $a\in \mathcal{R}_A$ is dense in $\mathcal{H}$. For the finite-dimensional case, say $d$-dimension, one could choose the state
  \bea\label{eqnow1}
  |\Psi\rangle:= \sum_{k=1}^d c_k |k\rangle_A \otimes |k\rangle_{\bar A},
  \eea
 with $c_k\ne 0$, $|k\rangle_{A(\bar A)}$ are basis of $\mathcal{H}_{A(\bar A)}$. It can be proved that $|\Psi\rangle$ is cyclic for the algebra $\mathcal{R}_A$ and $\mathcal{R}_{\bar A}$.
For the infinite-dimensional case, in the framework of algebraic QFT, one could also construct the local algebra $\mathcal{R}(A)$  consisting of the local operators supported in the open region $A$. The Reeh-Schlieder theorem \cite{RS-theorem} states that the vacuum state $|0\rangle$ is cyclic for the algebra $\mathcal{A}$ associated with any bounded open region $A$. Therefore, one could construct any pure state $|\psi\rangle$ in $\mathcal{H}$ by only using the operators in $\mathcal{R}(A)$ or $\mathcal{R}(\bar A)$. In Appendix \ref{algebricreview}, we give a brief review of the associated aspects of algebraic QFT, see also \cite{Witten:2018zxz, Haag:91sp}.

Assume the state $|\Psi\rangle$ is cyclic for the algebra $\mathcal{R}(A)$, there exist $a$ and $\tilde a$ such that
\bea\label{statedefined}
a|\Psi\rangle =|\psi\rangle,\quad \tilde{a}|\Psi\rangle =|\phi\rangle.
\eea
Substituting the above expression into Eq.(\ref{Tpseudo}), the $\eta$-pseudo-Hermitian transition matrices can be expressed as
\bea\label{general}
\mathcal{T}^{\psi|\phi}=\frac{a|\Psi\rangle \langle \Psi| a^\dagger \eta}{\langle \Psi| a^\dagger \eta a|\Psi\rangle},
\eea
where $\eta$, according to Proposition \ref{propo2}, is the direct product of two invertible Hermitian matrices $\eta_{A}$ and $\eta_{\bar A}$.
\subsection{Finite-dimensional examples}
 For simplicity, in the following, we choose $c_k$ in \eqref{eqnow1} to be positive. One can assume the operators $a$, $\eta_A$, and $\eta_{\bar A}$ in Eq.\eqref{general} are of the form $a=\sum_{ij}a_{ij}|i\rangle_A ~_A\langle j|$, $\eta_{A}=\eta_{mn}|m\rangle_A ~_A\langle n|$
 and $\eta_{\bar A}=\bar{\eta}_{mn}|m\rangle_{\bar A} ~_{\bar A}\langle n|$, where the matrices $\eta_{mn}$ and $\bar{\eta}_{mn}$ are invertible and Hermitian. Then we construct the transition matrix $\mathcal{T}^{a}$ by using the formula (\ref{general}) as follows
 \bea\label{FiniteTA}
 \mathcal{T}^a_{A}=\mathcal{N} \sum_{\substack{j,k\\i',j',k'}}a_{jk'}c_{k'} a^*_{j'i'}c_{i'}\eta_{j'k}\bar{\eta}_{i'k'}|j\rangle_A ~_A\langle k|,
 \eea
where $\mathcal{N}$ is the normalization. We present some numerical $\eta$-pseudo-Hermitian transition matrices in Appendix \ref{app:E.1}.
\subsection{Examples in QFTs}\label{ExampleinQFT}
\subsubsection{2-dimensional rational CFTs}
The first  example in QFTs is the real-time evolution of pseudo  R\'enyi entropy in 2-dimensional CFTs considered in \cite{Guo:2022sfl}.
The subsystem $A$ is taken to be $[-L,L]$ ($L>0$) \footnote{Note that $A$ is chosen to be $[0,L]$ in \cite{Guo:2022sfl}, which, according to the spatial translation symmetry, does not affect the final conclusion. }. Consider the transition matrix
\bea\label{transitionCFT}
\mathcal{T}^{\mo}:=\frac{ \mo(x_1,t_1)|0\rangle \langle 0|\mo(x_2,t_2)}{\langle 0| \mo(x_2,t_2)\mo(x_1.t_1)|0\rangle},
\eea
where $\mo$ is assumed to be the Hermitian primary operator.
Consider the case with {$t_1=t_2=-t$} and $x_1=-x_2$.  We find for some rational CFT models $\tr[(\mathcal{T}^{\mathcal{O}}_A)^2]$ is always real \cite{Guo:2022sfl}. Assume the Hamiltonian $H$ commutes with the parity $\mathrm{P}$. The transition matrix can be written as
 \bea
\mathcal{T}^\mo=\frac{\mo(x_1,-t)|0\rangle \langle 0| \mo^\dagger(x_1,-t)\mathrm{P}}{\langle 0| \mo^\dagger(x_1,-t)\mathrm{P} \mo(x_1,-t)|0\rangle}.
 \eea
 It shows that $\mathcal{T}^\mo$ is $\mathrm{P}$-pseudo-Hermitian according to Eq.\eqref{Tpseudo}. Since the subsystem $A$ and $\bar A$ are invariant under parity $\mathrm{P}$, we could decompose $\mathrm{P}=\mathrm{P}_A\otimes \mathrm{P}_{\bar A}$, where $\mathrm{P}_{A (\bar A)}$ acts on $A(\bar A)$. We can schematically write $|\phi(x)\rangle$ ($x\in A$) and $|\phi(\bar x)\rangle$ ($\bar x \in \bar A$) as the basis for subsystem $A$ and $\bar A$. The action of $\mathrm{P}_{A(\bar A)}$ is regarded as $\mathrm{P}_A |\phi(x)\rangle= |\phi(-x)\rangle$. By definition $-x \in A$ if $x\in A$. $\mathrm{P}_A$ maps the basis of $A$ into itself.  The parity operators are invertible. Therefore, $\mathcal{T}^\mo_{A(\bar A)}$ is $\mathrm{P}_{A(\bar A)}$-pseudo-Hermitian. This leads to $\tr[ (\mathcal{T}^\mo_{A(\bar A)})^2]$ is real, which is consistent with the results in \cite{Guo:2022sfl}. Again, our results in this paper predict $\tr[(\mathcal{T}^\mo_{A(\bar A)})^n]$ should be real for any positive integer $n\ge 2$.
 \subsubsection{General QFTs in Minkowski space}
 Our next example aims to demonstrate the construction of a transition matrix generating non-negative pseudo R\'enyi entropies within the context of general QFT in Minkowski space. We will see that it is closely related to Tomita-Takesaki theory.\footnote{We briefly review the Tomita-Takesaki theory and its application in algebraic QFTs in Appendix \ref{algebricreview}.} Consider a QFT that resides in a $d$-dimensional Minkowski space $\mathcal{R}_{1,d-1}$. The metric on   $\mathcal{R}_{1,d-1}$ is given by $ds^2=-dt^2+dx^2+d\vec{y}^2$, where $\vec{y}$ are coordinates of $(d-2)$-dimensional Euclidean space. We choose the subsystem $A$ to be the half-space $\{(x,\vec{y})|x<0\}$ and $\bar A$ its complement. The local algebra $\mathcal{R}_A$  is given by operators located at the left Rindler wedge $\mathcal{W}_A:=\{(t,x,\vec{y})|x< -|t|\}$, which is the causal domain of $A$.
The operators can be constructed by the smeared field $\int d^dx f(x^\mu)\phi(x^\mu)$ with the functions $f$ supported in $\mathcal{W}_A$. Similarly, the algebra $\mathcal{R}_{\bar A}$ is associated with the right wedge $\mathcal{W}_{\bar A}:= \{(t,x,\vec{y})|x>|t| \}$.

There exists an antiunitary operator $J_\Omega$, called modular conjugation, that exchanges the algebras $\mathcal{W}_A$ and $\mathcal{W}_{\bar A}$ according to the Tomita-Takesaki theory. For a given Hermitian operator $\phi(t,x,\vec{y})$, $J_\Omega$ acts as
\bea\label{eq451}
J_\Omega \phi(t,x,\vec{y})J_\Omega =\phi(-t,-x,\vec{y}).
\eea
It has been proved that $J_\Omega=\mathrm{CRT}$, where $\mathrm{C}$ and $\mathrm{T}$ are charge and time reversal operators, and $\mathrm{R}$ is the reflection $x\to -x$ while keeping other coordinates invariant \cite{Witten:2018zxz,Bisognano:1976za}. According to Eq.\eqref{eq451}, for any operator $\mathcal{O}_A\in \mathcal{R}_A$ we can define the operator $\mathcal{O}_{\bar A}:= J_\Omega\mathcal{O}_AJ_\Omega \in \mathcal{R}_{\bar A}$. For any pure state $|\psi\rangle$ there exists $\mathcal{O}_A\in \mathcal{R}_A$ such that $|\psi\rangle$ can be approximated by $\mathcal{O}_A|0\rangle$ by the cyclic property of the vacuum state $|0\rangle$ for the algebra $\mathcal{R}_A$.
Define the transition matrix
\bea\label{QFTexample}
\mathcal{T}^{\mathcal{O}_A}= \frac{\mathcal{O}_A|0\rangle \langle 0|\mathcal{O}_{\bar A} }{\langle 0|\mathcal{O}_A \mathcal{O}_{\bar A} |0\rangle },
\eea
where $\mathcal{O}_{\bar A}= J_\Omega\mathcal{O}_A J_\Omega$. It can be shown the spectra of $\mathcal{T}^{\mathcal{O}_A}_A$ are non-negative. Thus the pseudo R\'enyi entropy is non-negative. To show this result, we need to use the modular theory of QFTs.  By using the modular theory
\bea
\mathcal{T}^{\mathcal{O}_A}=\frac{\mathcal{O}_A|0\rangle\langle 0| \mathcal{O}_{A}^\dagger \Delta_\Omega^{1/2}}{\langle 0| \mathcal{O}_A^\dagger \Delta^{1/2}_\Omega \mathcal{O}_{A}|0\rangle},
\eea
where $\Delta_{\Omega}$ is the modular operator, a positive Hermitian operator. $\mathcal{T}^{\mathcal{O}_A}$ takes the same form as (\ref{general}). Thus it is $\Delta_\Omega^{1/2}$-pseudo-Hermitian. Further, $\Delta_\Omega^{1/2}=e^{-\pi K_A}\otimes e^{\pi K_{\bar A}}$. By using Theorem \ref{corol3} and the fact that $e^{-\pi K_A}=e^{-\pi K_A/2} e^{-\pi K_A/2}$ and $e^{\pi K_{\bar A}}=e^{\pi K_{\bar A}/2}e^{\pi K_{\bar A}/2}$ we conclude the eigenvalues of $\mathcal{T}^{\mathcal{O}_A}_A$ are all positive. In Appendix \ref{detailQFT} we show more general examples with positive pseudo R\'enyi entropy.

One could check the above result by evaluating the pseudo R\'enyi entropy by QFT methods. To move on, let's focus on 2D CFTs \footnote{In Appendix \ref{app:replica}, we give an overview of the replica method to compute the $n$th pseudo R\'enyi entropy in 2D CFTs.}. Consider the transition matrix (\ref{transitionCFT}).
 If $x_1=-x_2$ and $t_1=-t_2$, $J_\Omega \mo(x_1,t_1)J_\Omega=\mo(x_2,t_2)$, thus the pseudo R\'enyi entropy is expected to be real in this case. The results are shown in Figure \textcolor{black}{\ref{pic:Ising}} for models in rational CFTs.
  \begin{figure}[t]
\centering
\includegraphics[width=9cm]{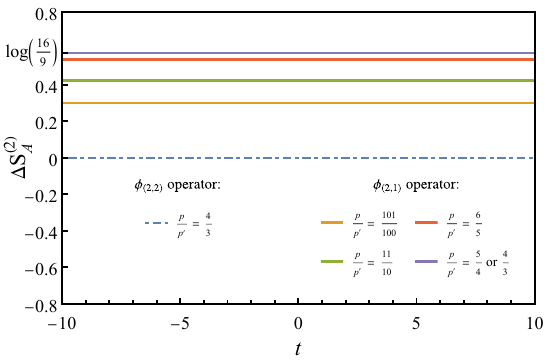}
\caption{The excess of the 2nd pseudo R\'enyi entropy $\Delta S^{(2)}_A$ ($\Delta S^{(2)}_A\equiv S^{(2)}_A-S^{(2)}_{A;vac}$, where $S^{(2)}_{A;vac}$ denotes the 2nd R\'enyi entropy of $A$ when the total system is in the vacuum) of the transition matrix $\T_A\equiv\tr_{\bar{A}}\frac{\mathcal{O}(x,t)|0\rangle\langle0|\mathcal{O}(-x,-t)}{\langle0|\mathcal{O}(-x,-t)\mathcal{O}(x,t)|0\rangle}$ in the minimal models $\mathcal{M}(p,p')$. We study the case of $\mo=\phi_{(2,2)}$ (dot-dashed line) and $\mo=\phi_{(2,1)}$ (solid line), respectively.  One novel feature is that the 2nd pseudo entropy is real and  time-independent. }\label{pic:Ising}
\end{figure}
\subsection{\textcolor{black}{Comment on an exceptional case in the 2-dimensional  free scalar theory} }
In \cite{Nakata:2020luh}, the authors studied the pseudo-R\'enyi entropy of locally excited states in a massless scalar field theory that resides in  two-dimensional Euclidean space.\footnote{In Euclidean space the coordinates are $\tau,x$ or $w=x-i\tau$, $\bar w=x+i\tau$, and the action is given by $S=\int dw d\bar w \partial_w \phi \partial_{\bar w} \phi$.} The  states of interest $|\psi\rangle$ and $|\phi\rangle$ are chosen as
\begin{align}\label{freescalarstates}
|\psi\rangle=e^{-a H} \tilde{\mathcal{O}}(x=0)|0\rangle,\quad
|\phi\rangle=e^{-a' H} \mathcal{O}(x=0)|0\rangle,
\end{align}
where $H$ is the Hamiltonian of CFT, $a$ and $a'$ are cutoff to avoid UV divergence. The inserted operators $\mo$ and $\tilde{\mathcal{O}}$ are defined as
\bea\label{exampleCFT1operator}
\mathcal{O}=e^{\frac{i}{2}\phi}+e^{-\frac{i}{2}\phi},\quad \tilde{\mathcal{O}}:= e^{\frac{i}{2}\phi}+e^{i\theta} e^{-\frac{i}{2}\phi},\quad\big(\theta\in [-\pi,\pi]\big).
\eea
 In \cite{Nakata:2020luh}, the authors showed that the 2nd and 3rd pseudo R\'enyi entropies of $\mathcal{T}^{\psi|\phi}_A$ are real for  subsystem $A$ with finite length, utilizing the replica method and conformal mapping.

To clarify the reality of the pseudo R\'enyi entropy within our framework, it is necessary to construct a transition matrix  that satisfies condition 2 in Theorem \ref{corol3}. However, it should be noted that in this particular case, proving the pseudo-Hermiticity of the transition matrix $\mathcal{T}^{\psi|\phi}$ is a challenging task, as we cannot directly apply Proposition \ref{propoforph} from finite-dimensional quantum mechanics to QFT.

Let us first consider the case of $\theta\neq0$ to show that $\mathcal{T}^{\psi|\phi}$  is pseudo-Hermitian. One first defines the momentum operator \cite{Blumenhagen:2009zz}
\bea{\label{momentumfree}}
\pi_0:= \frac{i}{4\pi}\int dx  \partial_\tau \phi,
\eea
which satisfies the following commutation relation
\bea{\label{commpiV}}
[\pi_0, \mathcal{V}_\alpha]=\alpha \mathcal{V}_\alpha,\quad(\mathcal{V}_\alpha\equiv:e^{i\alpha\phi}:).
\eea
One can also show $\pi_0$ is a Hermitian operator and commutes with Hamiltonian $H$. By using the commutator (\ref{commpiV}) and the Baker–Campbell–Hausdorff formula, we have
\bea{}
e^{\lambda \pi_0} \mathcal{V}_\alpha e^{-\lambda \pi_0}=  e^{\alpha\lambda} \mathcal{V}_\alpha.
\eea
 The transition matrix  $\mathcal{T}^{\psi|\phi}$  then can be written as
\bea{\label{thetanezero}}
\mathcal{T}^{\psi|\phi}=\frac{|\Phi\rangle \langle \Phi| \eta_\Phi}{\langle \Phi|\eta_\Phi |\Phi\rangle},
\eea
with
\bea{\label{phitildeeta}}
\eta_\Phi:=e^{-i\theta \pi_0} e^{-(a'-a)H},\quad
|\Phi\rangle:= e^{-a H }(e^{\frac{i}{2}\phi(0)}e^{-\frac{i}{2}\theta}+e^{-\frac{i}{2}\phi(0)}e^{\frac{i}{2}\theta})|0\rangle.
\eea
We further introduce an operator $\tilde{\mathcal{T}}^{\psi|\phi}$
\begin{align}
\tilde{\mathcal{T}}^{\psi|\phi}:=\frac{(e^{\frac{i}{2}\phi(0)}e^{-\frac{i}{2}\theta}+e^{-\frac{i}{2}\phi(0)}e^{\frac{i}{2}\theta})|0\rangle \langle 0| (e^{\frac{i}{2}\phi(0)}e^{-\frac{i}{2}\theta}+e^{-\frac{i}{2}\phi(0)}e^{\frac{i}{2}\theta})}{2\cos\frac{\theta}{2}{\langle e^{\frac{i}{2}\phi(0)}e^{-(a+a')H}e^{-\frac{i}{2}\phi(0)}\rangle}} e^{-i\theta\pi_0},\quad\mathcal{T}^{\psi|\phi}= e^{-a H} \tilde{\mathcal{T}}^{\psi|\phi} e^{-a'H}.\label{transionunitary}
\end{align}
Since the momentum operator $\pi_0$ is Hermitian, $e^{-i\theta \pi_0}$ is a unitary operator. Then, $\eta_\Phi$ is not a Hermitian operator.
The transition matrix (\ref{thetanezero}) seems to be different from the general form
of the pseudo-Hermitian operator (\ref{Tpseudo}). It doesn't mean the transition matrix cannot be pseudo-Hermitian.
The transition matrix $\tilde{\mathcal{T}}^{\psi|\phi}\propto|\Psi\rangle \langle \Psi|U$ with the unitary operator $U:=e^{-i\theta\pi_0}$ and the state $|\Psi\rangle:=(e^{\frac{i}{2}\phi(0)}e^{-\frac{i}{2}\theta}+e^{-\frac{i}{2}\phi(0)}e^{\frac{i}{2}\theta})|0\rangle$. $\tilde{\mathcal{T}}^{\psi|\phi}$ can be $\eta$-pseudo-Hermitian if the unitary operator $U$ and the state $|\Psi\rangle$ satisfy certain constraints.
It is necessary that $\eta$, $U$, and $|\Psi\rangle$ should satisfy that $U \eta|\Psi\rangle=|\Psi\rangle$.\footnote{Please refer to Appendix \ref{Detailoffreescalar} for more details.} Equivalently, the operator $\eta':=U\eta$ should satisfy the condition (\ref{cond1})  (\ref{cond2})  or a weaker condition (\ref{constraintsetaU}).
One could show the transition matrix $\mathcal{T}^{\psi|\phi}$ is $\eta e^{-(a'-a)H}$-pseudo-Hermitian if $\tilde{\mathcal{T}}^{\psi|\phi}$ is $\eta$-pseudo-Hermitian and $\eta$ commutes with $H$.  Assume $\tilde{\mathcal{T}}^{\psi|\phi}$ is $\eta$-pseudo-Hermitian, we have
\bea{\label{tildeeta}}
\eta\tilde{\mathcal{T}}^{\psi|\phi}\eta^{-1}=\tilde{\mathcal{T}}^{\psi|\phi\dagger}.
\eea
By using the above result and the definition \eqref{transionunitary} one could show
\bea{}
&&(\eta e^{-(a'-a)H)})\mathcal{T}^{\psi|\phi}(\eta e^{-(a'-a)H)})^{-1}\nn \\
&&=\eta e^{-a'H}\tilde{\mathcal{T}}^{\psi|\phi}e^{-a H}\eta^{-1}\nn \\
&&=e^{-a'H} (\tilde{\mathcal{T}}^{\psi|\phi})^\dagger e^{-aH}={\mathcal{T}^{\psi|\phi}}^\dagger,
\eea
where in the second step we use $\eta$ commutes with $H$ and (\ref{tildeeta}). It is expected that there exists an operator $\eta'$ satisfying the condition (\ref{constraintsetaU}), or equally $\eta' \pi_0 \eta'=-\pi_0 $. A candidate for $\eta'$ is the time reflection operator $\Theta$ \cite{Osterwalder:1973dx,Osterwalder:1974tc}, which gives the transformation $\Theta \phi(\tau,x)\Theta =\phi(-\tau,x)$ and $ \Theta \partial_\tau \phi(\tau)\Theta =-\partial_\tau \phi(\tau)$. Thus we have $\Theta \pi_0\Theta =-\pi_0$. Further, one should note that $\Theta$ commutes with $H$.

We next consider the simpler case $\theta=0$. One can readily find that the transition matrix $\mathcal{T}^{\psi|\phi}$ in this case is  $e^{-(a'-a)H}$-pseudo-Hermitian.

It seems that our analysis can only come to a  standstill here.  If we want to continue applying Theorem \ref{corol3}, we are faced with the problem of decomposing $e^{-(a'-a)H}$ into the tensor product of operators in $A$ and $\bar A$, regardless of whether $\theta$ is zero or not. A naive speculation could be $e^{-(a'-a)H}=e^{-(a'-a)H_{A}}\otimes e^{-(a'-a)H_{\bar A}}$, where $H_{A(\bar A)}=\int_{A(\bar A)} dx T_{00}$. \textcolor{black}{However, if one takes the QFT as a lattice model, this decomposition would be ambiguous.  It seems there exists some interaction term between $A$ and $\bar A$\footnote{We thank the anonymous referee for pointing this out to us.}.} We could consider a lattice-regularized massless free scalar to see this clearly. The lattice-regularized Hamiltonian of it is given by
\begin{align}
H=\sum_{n=-\infty}^{+\infty}\frac{1}{2}a\left\{\dot{\phi}_n^2+\frac{1}{a^2}(\phi_{n+1}-\phi_n)^2\right\},\quad (\dot{\phi}_n\equiv\frac{d}{dt}\phi_n),\label{explain H}
\end{align}
where $a$ is the lattice size which is so-called regulator.
From Eq.\eqref{explain H}, it is evident that the term $(\phi_{n+1}-\phi_{n})^2/a^2$ contributes to the interaction term at the boundary between $A$ and $\bar A$. For instance, in the case where $A=[0,\infty)$, the interaction term takes the form $H_{\text{int}}=\frac{1}{2a}(\phi_{0}-\phi_{-1})^2$, where $\phi_0\equiv\phi(x=0)\in A,~\phi_{-1}\equiv\phi(x=-a)\in\bar{A}$. Thus, it cannot be decomposed $H$ into $H_A\otimes H_{\bar A}$ naively in the lattice model. \textcolor{black}{In the continuous limit $a\to 0$, it seems the interaction term becomes a boundary term. But it is hard to estimate whether this boundary term would be important for our purpose. For the lattice model with interactions, the problem would be more subtle. }

On the other hand, we would like to emphasize that the $\eta$ matrix that ensures the transition matrix satisfies the pseudo-Hermiticity condition $\eqref{phc}$ is not unique. Our results only provide one possibility, and \textcolor{black}{the question of whether a decomposable $\eta$ matrix exists is an intriguing one that deserves further investigation.}

\section{Discussion}
Real-valued pseudo entropies have robust potential applications in holography, quantum information, and quantum many-body physics. In this article, we explore the reality condition of pseudo entropy utilizing the notion of pseudo-Hermiticity. \textcolor{black}{We have extended the concept of transition matrices from pure states to more general mixed states \eqref{defofX}, and have demonstrated that all finite-dimensional transition matrices are pseudo-Hermitian (see Proposition \ref{propoforph}).} We derive the equivalent condition for the reduced transition matrix to be pseudo-Hermitian (see Proposition \ref{propo2}). On this basis, we construct a class of transition matrices that are $\eta$-pseudo-Hermitian and  have non-negative eigenvalues (not all zeros) by selecting special $\eta$ matrices (see Theorem \ref{corol3} about pure state transition matrices and
Corollary \ref{corol4} about mixed state transition matrices).  Thus, we construct a set of transition matrices that  can generate non-negative pseudo (R\'enyi) entropy of arbitrary order of  subsystems.

Based on our constructions, we perform a series of numerical tests in finite dimensional quantum systems. We subsequently work out several non-trivial examples in quantum field theories, including Euclidean QFT with non-Hermitian operator insertion, and
QFTs with Lorentzian  signature. The core idea is to find an appropriate $\eta$ matrix so that the transition matrix is $\eta$-pseudo-Hermitian and satisfy the conditions of Theorem \ref{corol3}.  In terms of these examples, we learn that the $\eta$-pseudo-Hermiticity is associated with the various aspects, e.g., algebraic structures of QFT, modular Hamiltonian, parity, etc.


The notion of pseudo-Hermiticity has been extensively studied for the non-Hermitian systems. It originates from non-Hermitian matrices, which are diagonalizable and own a complete biorthonormal eigenbasis. In the current work, we mainly focus on the pseudo-Hermiticity of the reduced transition matrix to construct the real-valued pseudo entropy. If the reduced transition matrix is non-pseudo-Hermitian which means it is not diagonalizable, one possible way to find the reality condition of pseudo entropy is writing the matrix into Jordan form for the non-diagonalizable cases \cite{Mostafazadeh:2002wg,Ashida:2020dkc}. It is an interesting direction to achieve the real-valued pseudo entropy by going beyond the pseudo-Hermiticity.
\subsection*{Acknowledgements}
We would like to thank Jun-Peng Cao, Pak Hang Chris Lau and Long Zhao for valuable discussions related to this work.
WZG is supposed by the National Natural Science Foundation of China under Grant No.12005070 and the Fundamental Research
Funds for the Central Universities under Grants NO.2020kfyXJJS041. SH would like to appreciate the financial support from Jilin University, Max Planck Partner Group, and the Natural Science Foundation of China Grants (No.12075101, No.12235016).
\appendix
\section{ Non-pseudo-Hermitian $X_A$ with real $\tr[(X_A)^n]$}\label{feiduijiao}
Let us consider a 4 qubits system $\mathcal{S}$ (2 qubits each for $A$ and $\bar{A}$) and a transition matrix $X=\frac{|\psi\rangle\langle\phi|}{\langle\phi|\psi\rangle}$ acting on its Hilbert space $\mathcal{H_S}\equiv H_A\otimes H_{\bar{A}}$, where $|\psi\rangle$ and $|\phi\rangle$ are two non-orthogonal quantum states living in $H_{\mathcal{S}}$,
\begin{align}
|\psi\rangle=&\frac{1}{2}|00\rangle_A|00\rangle_{\bar{A}}+\frac{1}{2}|01\rangle_A|01\rangle_{\bar{A}}+\frac{1}{2}|10\rangle_A|10\rangle_{\bar A}+\frac{1}{2}|11\rangle_A|11\rangle_{\bar A},\nn\\
|\psi_{\perp}\rangle=&\frac{i}{4}|00\rangle_A|00\rangle_{\bar{A}}+\frac{i}{4}|01\rangle_A|01\rangle_{\bar{A}}-\frac{i}{4}|10\rangle_A|10\rangle_{\bar A}-\frac{i}{4}|11\rangle_A|11\rangle_{\bar A}+\frac{\sqrt{3}}{2}|11\rangle_A|10\rangle_{\bar{A}},\nn\\
|\phi\rangle\equiv&\frac{\sqrt{2}}{2}|\psi\rangle+\frac{\sqrt{2}}{2}|\psi_{\perp}\rangle,\quad\quad(\langle\phi|\phi\rangle=\langle\psi|\psi\rangle=\langle\psi_{\perp}|\psi_{\perp}\rangle=1,~\langle\psi|\psi_{\perp}\rangle=0).
\end{align}
The reduced transition matrix of the subsystem $A$, obtained by tracing out the d.o.f. of $\bar{A}$, is given by
\begin{align}
X_A\equiv&\tr_{\bar{A}}X\nn\\
=&\Big(\frac{1}{4}-\frac{i}{8}\Big)\big(|00\rangle_A\mathbin{_A\langle}00|+|01\rangle_A\mathbin{_A\langle}01|\big)+\Big(\frac{1}{4}+\frac{i}{8}\Big)\big(|10\rangle_A\mathbin{_A\langle}10|+|11\rangle_A\mathbin{_A\langle}11|\big)+\frac{\sqrt{3}}{4}|01\rangle\langle11|.\label{app1:TA}
\end{align}
Building on \eqref{app1:TA}, it's more useful to write down the matrix formulation of $X_A$,
\be\label{matrixta}
X_A=\left(
       \begin{array}{cccc}
         \frac{1}{4}-\frac{i}{8} &   &   &   \\
           & \frac{1}{4}-\frac{i}{8} &   &   \\
           &   & \frac{1}{4}+\frac{i}{8} & \frac{\sqrt{3}}{4}  \\
           &   &   & \frac{1}{4}+\frac{i}{8} \\
       \end{array}
     \right),
\ee
which is an upper triangular $4\times4$ matrix and cannot be diagonalized. It can be found from \eqref{matrixta} that  the eigenvalues of $X_{A}$ consist of two complex conjugate pairs, which renders $\tr[(X_A)^n]$ real. On the other hand, we have
\be
X_A^{\dagger}=\left(
       \begin{array}{cccc}
         \frac{1}{4}+\frac{i}{8} &   &   &   \\
           & \frac{1}{4}+\frac{i}{8} &   &   \\
           &   & \frac{1}{4}-\frac{i}{8} &   \\
           &   & \frac{\sqrt{3}}{4}  & \frac{1}{4}-\frac{i}{8} \\
       \end{array}
     \right).
\ee
Although $X_A^{\dagger}$ has the same eigenvalues as $X_A$,  they are not similar. This is because they have different Jordan standard forms,  which is read from the fact that the Jordan blocks of the same eigenvalue of two matrices are different. Therefore, we know that $X_A$ is non-pseudo-Hermitian.
\section{\textcolor{black}{A proof of Theorem \ref{corol3}}}\label{proofoftheo3}
In this appendix, we show that $\eta_A^{1/2}\tilde{\mathcal{T}}_A\eta_A^{1/2}$ is positive semi-definite and not null in the following four cases: $\mathfrak{a}$) $\eta_A$ and $\eta_{\bar{A}}$ are both positive definite; $\mathfrak{b}$) $\eta_A$ and $\eta_{\bar{A}}$ are both negative definite; $\mathfrak{c}$) $\eta_A$ is positive definite and $\eta_{\bar{A}}$ is negative definite; $\mathfrak{d}$) $\eta_A$ is negative definite and $\eta_{\bar{A}}$ is positive definite.
\paragraph{Case $\mathfrak{a}$ :}  Since $\eta_A$ and $\eta_{\bar{A}}$ are both positive definite, we first have $\langle\psi|\eta|\psi\rangle=\langle\psi|\eta_A\otimes\eta_{\bar A}|\psi\rangle>0$ for any $|\psi\rangle\neq0$.\footnote{Note that the eigenvalues of $\eta$ consist of  the product of the respective eigenvalues of $\eta_A$ and $\eta_{\bar A}$.} We also know that $\eta_A^{1/2}$ and $\eta_{\bar A}^{1/2}$ are both positive definite. Therefore, we find that $\tilde{\mathcal{T}}_{A}$ \eqref{ttaat} is positive semi-definite and not null.\footnote{Note that $\tilde{\mathcal{T}}_{A}$ can have zero eigenvalues, depending on the choice of $|\psi\rangle$.} It follows that $\eta_A^{1/2}\tilde{\mathcal{T}}_A\eta_A^{1/2}$ is positive semi-definite and not null.

\paragraph{Case $\mathfrak{b}$ :} We first have $\langle\psi|\eta|\psi\rangle>0$ for any $|\psi\rangle\neq0$. Since $(-\eta_{\bar A})^{1/2}$ is positive definite in this case, $\tilde{\mathcal{T}}_{A}$ can be rewritten as
\begin{align}
\tilde{\mathcal{T}}_{A}=\frac{\tr_{\bar A}\left[ |\psi\rangle \langle \psi| \eta_{\bar A}\right]}{\langle \psi| \eta |\psi\rangle}=-\frac{\tr_{\bar A}\left[(-\eta_{\bar A})^{1/2} |\psi\rangle \langle \psi| (-\eta_{\bar A})^{1/2}\right]}{\langle \psi| \eta |\psi\rangle}\label{eq52ta}.
\end{align}
 From above we conclude that $\tilde{\mathcal{T}}_{A}$ is negative semi-definite and not null. On the other hand, the square root of $\eta_A$ can be rewritten as $\eta_A^{1/2}=i(-\eta_{A})^{1/2}$, where $(-\eta_{A})^{1/2}$ is positive definite.
 It follows that $\eta_A^{1/2}\tilde{\mathcal{T}}_A\eta_A^{1/2}=(-\eta_A)^{1/2}(-\tilde{\mathcal{T}}_A)(-\eta_A)^{1/2}$ is positive semi-definite and not null.

 \paragraph{Case $\mathfrak{c}$:} In the third case, we first know that $\langle\psi|\eta|\psi\rangle<0$ for any $|\psi\rangle\neq0$ and $(-\eta_{\bar A})^{1/2}$ is positive definite, which leads to $\tilde{\mathcal{T}}_{A}$ \eqref{eq52ta} being  positive semi-definite and not null. Since $\eta_{\bar A}^{1/2}$ is positive definite,  $\eta_A^{1/2}\tilde{\mathcal{T}}_A\eta_A^{1/2}$ is positive semi-definite and not null.

 \paragraph{Case $\mathfrak{d}$:} In the last case, we first have $\langle\psi|\eta|\psi\rangle<0$ for any $|\psi\rangle\neq0$, which leads to $\tilde{\mathcal{T}}_{A}$ \eqref{ttaat} being negative semi-definite and not null. Then $\eta_A^{1/2}\tilde{\mathcal{T}}_A\eta_A^{1/2}=(-\eta_A)^{1/2}(-\tilde{\mathcal{T}}_A)(-\eta_A)^{1/2}$ is positive semi-definite and not null.

In summary, we show that  $\eta_A^{1/2}\tilde{\mathcal{T}}_A\eta_A^{1/2}$ is always positive semi-definite and not null.
\section{A proof of Theorem \ref{corol2}}\label{corollary2}
The result of Theorem \ref{corol2} follows from the spectral decomposition of pseudo-Hermitian matrices. For any diagonalizable $\eta$-pseudo-Hermitian matrix $M$, we can write the spectral decomposition of $M$ as
\bea
M=\sum_{i} \lambda_{0,i} |\psi_{0,i} \rangle \langle \phi_{0,i} |+\sum_{j}\Big(\lambda_{+,j} |\psi_{+,j} \rangle \langle \phi_{+,j}|+\lambda_{-,j} |\psi_{-,j} \rangle \langle \phi_{-,j}|\Big),
\eea
where $\lambda$, $|\psi\rangle$ and $\langle\phi|$ represent the eigenvalue, right eigenvector, and left eigenvector of $M$, respectively \footnote{We use the subscript $0$ to stand for real eigenvalues and the corresponding basis eigenvectors and the subscript $\pm$ to stand for the complex eigenvalues with $\pm$ imaginary part and the corresponding basis eigenvectors.}. Since we can always choose the biorthonormal eigenbasis  such that
 \bea
 |\phi_{0,i}\rangle =\eta |\psi_{0,i}\rangle, \quad |\phi_{{\pm},j}\rangle=\eta| \psi_{{\mp},j}\rangle,\quad \langle\phi_{\alpha,i}|\psi_{\beta,j}\rangle=\delta_{\alpha\beta}\delta_{ij},~(\alpha,\beta\in\{0,\pm\})
 \eea
 hold \cite{Mostafazadeh:2001jk}, the spectrum decomposition becomes
 \begin{align}
 M =&\sum_{i} \lambda_{0,i} |\psi_{0,i} \rangle \langle \psi_{0,i} |\eta+\sum_{j}\left(\lambda_{+,j} |\psi_{+,j} \rangle \langle \psi_{-,j}|\eta+\lambda_{-,j}|\psi_{-,j} \rangle \langle \psi_{+,j}|\eta\right)\nonumber \\
 =&\sum_{i} \lambda_{0,i} |\psi_{0,i} \rangle \langle \psi_{0,i} |\eta \nonumber \\
 &+\sum_{j}\lambda^R_{+,j} \Big[(|\psi_{+,j} \rangle +|\psi_{-,j}\rangle)(\langle \psi_{-,j}|+  \langle \psi_{+,j}|)\eta-|\psi_{-,j}\rangle\langle \psi_{-,j}|\eta- |\psi_{+,j}\rangle\langle \psi_{+,j}|\eta\Big]\nonumber \\
 &+\sum_{j}\lambda^I_{+,j} \Big[(|\psi_{+,j} \rangle -i |\psi_{-,j}\rangle)(\langle \psi_{+,j}|+i \langle \psi_{-,j}|)\eta-|\psi_{-,j}\rangle\langle \psi_{-,j}|\eta- |\psi_{+,j}\rangle\langle \psi_{+,j}|\eta\Big]\nonumber\\
 =&\sum_{i} \lambda_{0,i} |\psi_{0,i} \rangle \langle \psi_{0,i} |\eta+\sum_{j}\lambda^R_{+,j} (|\psi_{+,j} \rangle +|\psi_{-,j}\rangle)(\langle \psi_{-,j}|+  \langle \psi_{+,j}|)\eta\nn\\
 &-\sum_j(\lambda^R_{+,j}+\lambda^I_{+,j})|\psi_{-,j}\rangle\langle \psi_{-,j}|\eta-\sum_j(\lambda^R_{+,j}+\lambda^I_{+,j}) |\psi_{+,j}\rangle\langle \psi_{+,j}|\eta\nonumber \\
 &+\sum_{j}\lambda^I_{+,j}(|\psi_{+,j} \rangle -i |\psi_{-,j}\rangle)(\langle \psi_{+,j}|+i \langle \psi_{-,j}|)\eta,
 \end{align}
 where $\lambda_{+,j}^R$ and $\lambda_{+,j}^I$ are the real and imaginary  part of $\lambda_{+,j}$, respectively.
 Note that every term in summations is $\eta$-pseudo-Hermitian.
 \section{Brief review of modular theory in QFTs }\label{algebricreview}

For any given open subsystem $A$ in spacetimes, the local algebra $\mathcal{R}_A$ consists of all the operators supported in $A$. The algebra can also be associated with the domain of dependence of $A$, denoted by $\mathcal{D}(A)$. {The reason is that the operators located in $\mathcal{D}(A)$  can be determined by the ones in $A$ according to the dynamical time evolution of the theory. If $A'$ is another subsystem that is spacelike with $A$, we expect the operators in $A'$ would commute with the ones in $A$, that is $[\mathcal{R}_A,\mathcal{R}_{A'}]=0$. }

Denote the algebra associated with the whole spacetime as $\mathcal{R}$.
The full Hilbert space $\mathcal{H}_0$ could be constructed by acting the operators in $\mathcal{R}$ on the vacuum state $|0\rangle$. The Reeh-Schlieder theorem says that the set $\{ a|0\rangle, a\in \mathcal{R}_A\}$  is also dense in $\mathcal{H}_0$. For any given state $|\psi\rangle$, the theorem means that there exist operator $a \in \mathcal{R}_A$ such that $a|0\rangle$ can be arbitrarily close to $|\psi\rangle $. Thus we could construct the transition matrix $\mathcal{T}^{\psi|\phi}$ only by using the operators located in a subsystem. The above results can also be generalized to any cyclic state $|\Psi\rangle$.

The Tomita operator $S_\Psi$ for the state $|\Psi\rangle$ is antilinear and satisfies
\bea\label{Smodular}
S_\Psi a|\Psi\rangle =a^\dagger |\Psi\rangle,
\eea
for any $a\in \mathcal{R}_A$. By definition it is obvious that $S_\Psi^2=1$. $S_\Psi$ has a unique polar decomposition
\bea
S_\Psi=J_\Psi \Delta_\Psi^{1/2},
\eea
where $J_\Psi$ is antiunitary, $\Delta_\Psi^{1/2}$ is a positive Hermitian operator. $J_\Psi$ is called the modular conjugation satisfying $J_\Psi^2=1$ and $J_\Psi^\dagger =J_\Psi$. $\Delta_\Psi $ is the modular operator associated with $\mathcal{R}_A$ and $|\Psi\rangle$. Similarly, one could define the modular operator $\bar S_{\Psi}$ associated with $\mathcal{R}_{\bar A}$. By using $S_\Psi^2=1$ we have
\bea\label{Jdeltarelation}
J_\Psi \Delta_\Psi^{1/2} J_\Psi=\Delta_{\Psi}^{-1/2}.
\eea
It can be shown that
\bea
\bar S_\Psi=S_\Psi^\dagger=\Delta_\Psi^{1/2} J_\Psi=J_\Psi \Delta^{-1/2}.
\eea

Consider the $d$-dimensional Minkowski spacetime. The metric is $ds^2=-dt^2+dx^2+d\vec{y}^2$, where $\vec{y}$ are coordinates of $(d-2)$-dimensional Euclidean space. Let the subsystem $A$ be $x>0$. The domain of dependence of $A$ is known as the Rindler wedge $\mathcal{W}_A$. For the vacuum state $|0\rangle$ the modular conjugation $J_\Omega$ is given by
\bea
J_\Omega=\mathrm{CRT},
\eea
which is first proved by Bisognano and Wichmann \cite{Bisognano:1976za}. The modular operator $\Delta_\Omega$ can be formly written as
\bea
\Delta_\Omega=\rho_{A}\otimes \rho^{-1}_{\bar A},
\eea
where $\rho_{A}:= e^{-2\pi K_A}$ and $\rho_{\bar A}:= e^{-2\pi K_{\bar A}}$ are the reduced density matrices of $A$ and $\bar A$. $K_A$ and $K_{\bar A}$ are known as the modular Hamiltonian of $A$ and $\bar A$. The density matrices are positive operators. For any positive function $f(x)$, the operators $f(\rho_A)$ or $f(\rho_{\bar A})$ are also positive. For example, one could define the operator $\rho_A^{1/2}= e^{-\pi K_A}$, $\rho_A^{1/4}= e^{- \pi K_A/2}$.
It is obvious the modular operator $\Delta_{\Omega}=e^{-2\pi (K_A-K_{\bar A})}$ is a positive Hermitian operator.

 For the Rindler wedge $K_A$ and $K_{\bar A}$ are associated with the Lorentz  boost generators
\bea
&&K_A=\int_{t=0,x\ge 0}dx d^{d-2}y x T_{00},\nn\\
&&K_{\bar A}=-\int_{t=0,x\le 0}dx d^{d-2}y x T_{00}.
\eea
For any Hermitian operator $\mo(t,x,\vec{y})$, according to the definition of $S_{\Omega}$ we have
\bea
S_{\Omega}\mo(t,x,\vec{y})|0\rangle =\mo(t,x,\vec{y})|0\rangle,
\eea
which leads to
\bea\label{JDeltarelation}
\Delta_{\Omega}^{1/2}\mo(t,x,\vec{y})|0\rangle=J_{\Omega}\mo(t,x,\vec{y})|0\rangle|0\rangle =\mo(-t,-x,\vec{y})|0\rangle,\nn\\
~
\eea
where we use the fact $J_\Omega^2=1$ and $J_\Omega|0\rangle =|0\rangle$.
\section{Details of the example in QFTs}\label{detailQFT}

The transition matrix (\ref{QFTexample}) is related to the operators $\mathcal{O}_A$ and $\mathcal{O}_{\bar A}$. In Figure \ref{OAplot}, we show the positions of the two operators.
\begin{figure}[t]
\centering
\includegraphics[width=9.0cm]{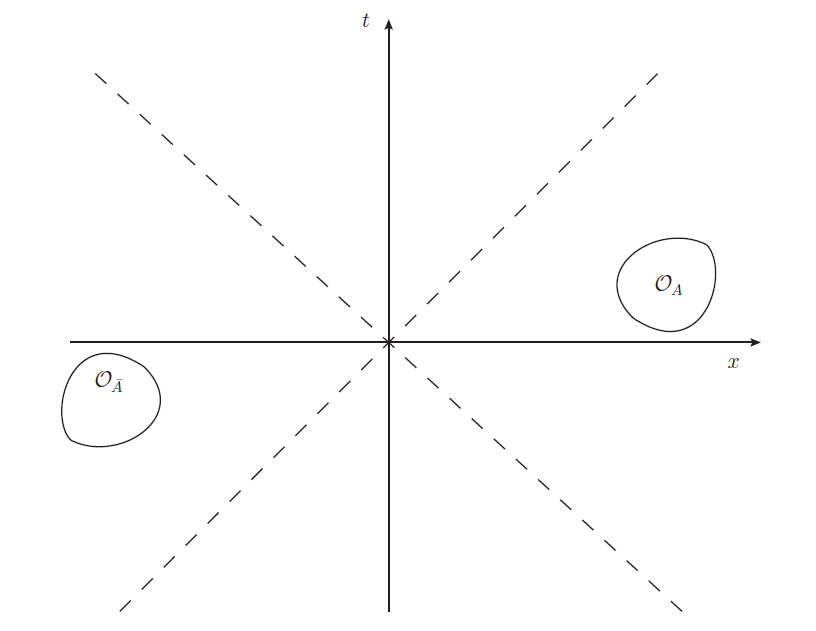}
\caption{Illustration of the operators $\mathcal{O}_A$ and $\mathcal{O}_{\bar A}$.}\label{OAplot}
\end{figure}
Taking $\mathcal{O}_{\bar A}$ into (\ref{QFTexample}) we obtain
\bea\label{step1}
\mathcal{T}^{\mathcal{O}_A}= \frac{\mathcal{O}_A|0\rangle \langle 0|\mathcal{O}_{ A}J_\Omega }{\langle 0|\mathcal{O}_A J_\Omega\mathcal{O}_{ A} |0\rangle }.
\eea
By the definition of Tomita operator we have
\bea
&&S_{\Omega}\mathcal{O}^\dagger_{A}|0\rangle=J_\Omega \Delta_\Omega^{1/2}\mathcal{O}_{A}|0\rangle\nn \\
&&\phantom{S_{\Omega}\mathcal{O}^\dagger_{A}|0\rangle}=J_\Omega \Delta_\Omega^{1/2}J_\Omega J_\Omega\mathcal{O}^\dagger_{A}|0\rangle\nn \\
&&\phantom{S_{\Omega}\mathcal{O}^\dagger_{A}|0\rangle}= \Delta_\Omega^{-1/2} J_\Omega\mathcal{O}^\dagger_{A}|0\rangle,
\eea
where in the second step we use the fact $J_\Omega^2=1$,in the third step we use (\ref{Jdeltarelation}). Therefore,  we have
\bea
J_\Omega \mathcal{O}_{A}^\dagger|0\rangle =\Delta_{\Omega}^{1/2}\mathcal{O}_A|0\rangle.
\eea
The transition matrix (\ref{step1}) is reduced to
\bea
\mathcal{T}^{\mathcal{O}_A}= \frac{\mathcal{O}_A|0\rangle \langle 0|\mathcal{O}^\dagger_{ A}\Delta_\Omega^{1/2} }{\langle 0|\mathcal{O}^\dagger_A \Delta_\Omega^{1/2}\mathcal{O}_{ A} |0\rangle }.
\eea
In the main text, we only discuss the special case (\ref{transitionCFT}), for which the eigenvalues of $\mathcal{T}^{\mathcal{O}_A}_A$ are positive real. More generally, one could choose $\mathcal{O}_A=\sum_j C_j \mo_j(x_1,t_1)$, where $\mo_j$ are Hermitian operators (not necessarily be primary), $C_j$ are arbitrary constants. By $\mathcal{O}_{\bar A}=J_\Omega \mathcal{O}_{\bar A} J_{\Omega}$ we have $\mathcal{\mo}_{\bar A}=\sum_j C_j^* \mo_j(-x_1,-t_1)$. It is expected that the pseudo R\'enyi entropy for the transition matrix associated with these operators will also be real.

The result of Theorem \ref{theor2} can be used to construct the $\eta$-pseudo-Hermitian transition matrix in QFTs. Assume $|\Psi\rangle$ is a cyclic state for the algebra $\mathcal{R}_A$. The general $\eta$-pseudo-Hermitian transition matrices in QFTs are
\bea
\mathcal{T}^{\mathcal{O}_A}=\frac{\mathcal{O}_A|\Psi\rangle \langle \Psi| \mathcal{O}_A^\dagger \eta}{\langle \Psi| \mathcal{O}_A^\dagger \eta \mathcal{O}_A|\Psi\rangle}.
\eea
If $\eta$ is taken to be identity, the transition matrix reduces to  the Hermitian case.
By using (\ref{Smodular}) one could rewrite the above formula as
\bea
&&\mathcal{T}^{\mathcal{O}_A}=\frac{\mathcal{O}_A|\Psi\rangle \langle \Psi| \mathcal{O}_A S_{\Psi}^\dagger \eta}{\langle \Psi| \mathcal{O}_A^\dagger \eta \mathcal{O}_A|\Psi\rangle}\nn\\
&&\phantom{\mathcal{T}^{\mathcal{O}_A}}=\frac{\mathcal{O}_A|\Psi\rangle \langle \Psi| \mathcal{O}_A J_\Omega \Delta_\Psi^{-1/2} \eta}{\langle \Psi| \mathcal{O}_A^\dagger \eta \mathcal{O}_A|\Psi\rangle}\nn\\
\eea
This paper only focuses on the vacuum state $|0\rangle$.
Our example (\ref{step1}) is a special case $\eta=\Delta_{\Omega}^{1/2}$. In general, one could choose $\eta=\eta_A\otimes \eta_{\bar A}$, where $\eta_{A}$ and $\eta_{\bar A}$ are invertible positive operators. Using Theorem \ref{corol3}, one could show that the pseudo R\'enyi entropy is also real in this case.
\section{Finite dimension example}
Assume the Hilbert space $\mathcal{H}=\mathcal{H}_A\otimes \mathcal{H}_{\bar A}$, the dimension of $\mathcal{H}_{A(\bar A)}$ is $d$.
(\ref{general}) provides us a way to generate the $\eta$-pseudo-Hermitian transition matrices with $\eta=\eta_{A}\otimes \eta_{\bar A}$. One could arbitrarily  choose the reference state $|\Psi\rangle$, i.e., the coefficients $c_k$ and the operators $a$, $\eta_{A(\bar A)}$. With a given basis $|k\rangle_{A(\bar A)}$, we have the expansion
\bea\label{matrices}
&&a=\sum_{ij}a_{ij}|i\rangle_A ~_A\langle j|,\nn\\
&&\eta_{A}=\sum_{m,n}\eta_{mn}|m\rangle_A ~_A\langle n|,\nn \\
&&\eta_{\bar A}=\sum_{ m, n}\bar{\eta}_{ m n}|m\rangle_{\bar A} ~_{\bar A}\langle n|.
\eea
The matrices $\eta_{mn}$ and $\bar \eta_{mn}$ should be Hermitian and invertible.

In finite dimension, it is easy to show the Reeh-Schlieder theorem. Any state $|\psi\rangle$ can be constructed by only local operations on $A$ or $\bar A$. The reference state $|\Psi\rangle:=\sum_k c_k |k\rangle_A \otimes |k\rangle_{\bar A}$ is cyclic if the coefficients $c_k$ are all non-vanishing. For any given state $|\psi\rangle$, we can expand it as
\bea
|\psi\rangle =\sum_{i,j} \psi_{ij}|i\rangle_A \otimes |j\rangle_{\bar A}.
\eea
It is enough to show that the basis  $|i\rangle_A\otimes |j\rangle_{\bar A}$ of $\mathcal{H}$ can be obtained only by local operations on $|\Psi\rangle$. One could achieve this by acting an operator $|i\rangle_{A}~_A\langle j|$ on $|\Psi\rangle$.

Taking (\ref{matrices}) into (\ref{general}) one could obtain the transition matrix $\mathcal{T}^a$. Eq.(\ref{FiniteTA}) can be obtained by partial trace $tr_{\bar A} \mathcal{T}^a:= \sum_k 	~_{\bar A}\langle k| \mathcal{T}^a |k\rangle_{\bar A}$.
One could generate random matrices $a_{ij}$, $\eta_{mn}$ and $\bar \eta_{mn}$ by software, e.g., Mathematica. Then we can construct the matrices $\mathcal{T}^a_A $ (\ref{FiniteTA}) and evaluate the eigenvalues of them.
 According to Corollary \ref{corol4},  the transition matrices by linear combinations of $\mathcal{T}^a$ can also have positive eigenvalues.

 We have the following three different cases.\\
\textit{Case I}: $\eta_{A(\bar A)}$ is Hermitian and invertible matrices.  \\
Generally, in this case, the eigenvalues are expected to come in complex conjugate pairs or be real. \\
\textit{Case II}: $\eta_{A(\bar A)}=\mo_{A(\bar A)}\mo^\dagger_{A(\bar A)}$. $\mo_{A(\bar A)}$ is an arbitrary invertible operator. \\
The eigenvalues, in this case, are expected to be real and positive. By considering the normalization of $\mathcal{T}^a_A$ the eigenvalues should belong to $[0,1]$. Thus the pseudo R\'enyi entropy should be real. \\
\textit{Case III}: The linear combinations of $\mathcal{T}^{a^I}$,
\bea
\mathcal{T}:= \sum_I x_I\mathcal{T}^{a^I},
\eea
where $x_I$ are positive numbers satisfying $\sum_I x_I =1$, $\mathcal{T}^{a^I}$ is $\eta_{A}\otimes \eta_{\bar A}$-pseudo-Hermitian transition matrices with  $\eta_{A(\bar A)}=\mo_{A(\bar A)}\mo^\dagger_{A(\bar A)}$.
In this case, the eigenvalues of $\mathcal{T}_A$ are positive.
\subsection{Numerical result with $d=3$}\label{app:E.1}
We show examples for these three cases in the main text, obtained by randomly choosing the matrices. In the following, we would like to show an example with $d=3$.\\
\textit{Case I}. The matrices $a_{ij}$, $\eta_{mn}$ and $\bar \eta_{mn}$ are randomly generated by  Mathematica,
\bea
&&\begin{footnotesize}
\eta_{A}
=\left(
\begin{array}{ccc}
 -12.7085 & 24.1113\, +2.50006 i & 12.752\, -7.64134 i \\
 24.1113\, -2.50006 i & -34.9796 & 16.159\, +12.3798 i \\
 12.752\, +7.64134 i & 16.159\, -12.3798 i & 6.06277 \\
\end{array}
\right)\nn
\end{footnotesize} \nn\\
&&\begin{footnotesize}
\eta_{\bar A}=\left(
\begin{array}{ccc}
 -18.2979& -5.89479-25.5118 i & 5.61273\, +21.1508 i \\
 -5.89479+25.5118 i & -32.9428& 12.504\, -10.931 i \\
 5.61273\, -21.1508 i & 12.504\, +10.931 i & -25.1785
\end{array}
\right)\nn
\end{footnotesize}  \nn\\
&&\begin{footnotesize}
a=\left(
\begin{array}{ccc}
 17.5055\, -19.3962 i & 8.29301\, +13.4073 i & -13.3458+5.79992 i \\
 -2.34212+12.1545 i & -19.0161+8.64625 i & 17.1027\, +20.3801 i \\
 -5.46605-24.0534 i & -0.924333+21.9112 i & -19.6201+22.0798 i \\
\end{array}
\right)\nn
\end{footnotesize}
\eea
The reference state $|\Psi\rangle =\frac{1}{\sqrt{3}}\sum_{k=1}^3|k\rangle_A|k\rangle_{\bar A}$. One could evaluate the reduced transition matrix $\mathcal{T}^a_{A}$ by using (\ref{FiniteTA}), it is given by
 \bea
\begin{footnotesize}
\mathcal{T}^a_{A}=\left(
\begin{array}{ccc}
 1.06475\, +0.82173 i & -2.58944-1.89593 i & -0.414598+1.47507 i \\
 5.81552\, +1.36514 i & -2.51542-3.45208 i & 3.90801\, +1.03913 i \\
 4.96095\, +0.0703348 i & -5.61733-0.740827 i & 2.45067\, +2.63035 i \\
\end{array}
\right)\nn
\end{footnotesize}
\eea
It is obvious that $\mathcal{T}^a_{A}$ is non-Hermitian. The eigenvalues of it are \bea
\lambda_1=0.7053\, -6.27836 i,\ \lambda_2=0.7053\, +6.27836 i,\ \lambda_3=-0.410601.
\eea
The pseudo R\'enyi entropy may not be  real in this case. e.g., $S^{(2)}(\mathcal{T}^a_{A})=-4.3525+3.14159 i$.\\
\textit{Case III}. We take $x_1=0.932007$, $x_2=0.0679932$. $\eta_{mn}$, $\bar \eta_{mn}$ and $a_{ij}$ are given by
\bea
&&\begin{footnotesize}
\eta_A=\left(
\begin{array}{ccc}
 2241.0 & -1009.73+735.915 i & 286.517\, +572.134 i \\
 -1009.73-735.915 i & 1007.3& 58.5703\, -617.441 i \\
 286.517\, -572.134 i & 58.5703\, +617.441 i & 1399.02 \\
\end{array}
\right)
\end{footnotesize}\nn \\
&& \begin{footnotesize}
\eta_{\bar A}=
\left(
\begin{array}{ccc}
 967.287 & -307.565-126.497 i & -129.349+149.126 i \\
 -307.565+126.497 i & 1336.39 & 209.52\, +1520.81 i \\
 -129.349-149.126 i & 209.52\, -1520.81 i & 2269.12 \\
\end{array}
\right)
\end{footnotesize}\nn\\
&& \begin{footnotesize}
a^1=
\left(
\begin{array}{ccc}
 -12.9325+0.0289028 i & -7.24499+4.48426 i & -14.4313+15.8304 i \\
 -4.09857-26.938 i & -14.2456-2.55161 i & 10.7265\, -5.71364 i \\
 4.30869\, +11.8775 i & -19.1378+9.46391 i & 1.32846\, +4.07899 i \\
\end{array}
\right)
\end{footnotesize}\nn\\
&& \begin{footnotesize}
a^2=
\left(
\begin{array}{ccc}
 -9.48366+25.7059 i & -6.14031+23.5242 i & -13.0021-20.8661 i \\
 -3.87512+5.57805 i & 4.9788\, -6.5475 i & 1.21723\, +7.54634 i \\
 -10.6898+13.5806 i & 11.563\, -1.35289 i & -14.61+21.6139 i \\
\end{array}
\right)
\end{footnotesize}\nn
\eea
The reference state $|\Psi\rangle =\frac{1}{\sqrt{3}}\sum_{k=1}^3|k\rangle_A|k\rangle_{\bar A}$. We have the reduced transition matrix $\mathcal{T}_A:=x_1 T^{a^1}_A+x_2 T^{a^2}_A$
\bea
\begin{footnotesize}
\mathcal{T}_A=\left(
\begin{array}{ccc}
 0.530706\, -0.0443678 i & -0.249067+0.220004 i & 0.16275\, +0.257015 i \\
 -0.0842444-0.152933 i & 0.129589\, +0.0451508 i & 0.0460981\, -0.13474 i \\
 0.213683\, -0.402251 i & 0.0403375\, +0.312084 i & 0.339705\, -0.000782992 i \\
\end{array}
\right)
\end{footnotesize}
\eea
The eigenvalues are
\bea
\lambda_1=0.938253,\ \lambda_2=0.0533309,\ \lambda_3=0.00841637.
\eea
The pseudo R\'enyi entropy is real. The result is shown in Figure \ref{case3plot}.
 \begin{figure}[h]
\centering
\includegraphics[width=9.0cm]{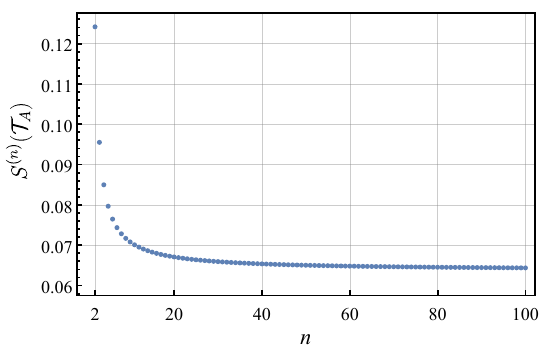}
\caption{The plot of $S^{(n)}(\mathcal{T}_A)$.}\label{case3plot}
\end{figure}
\subsection{Example: $S^{(n)}(\mathcal{T}_A)>0$, $\eta_A$ is not positive definite}\label{nonpositveexample}
Theorem \ref{corol3} only gives a sufficient condition for $
S^{(n)}(\mathcal{T}_A)>0$. In this section, we would like to use a numerical example to show it is not a necessary condition. We will focus on a three-dimensional example. Choosing the matrices
\bea
&&\begin{footnotesize}
\eta_A=\left(
\begin{array}{ccc}
 13.9359\, & -17.8554+8.22163 i & 0.740751\, -0.860494 i \\
 -17.8554-8.22163 i & 11.7561\,  & 3.87722\, +0.527719 i \\
 0.740751\, +0.860494 i & 3.87722\, -0.527719 i & 4.32501\,  \\
\end{array}
\right),
\end{footnotesize}\nn\\
&&
\begin{footnotesize}
\eta_{\bar A}=\left(
\begin{array}{ccc}
 2.68826\, +0. i & -2.76297+6.09204 i & -13.4254-5.89942 i \\
 -2.76297-6.09204 i & 23.4288\, +0. i & 2.24652\, -1.6307 i \\
 -13.4254+5.89942 i & 2.24652\, +1.6307 i & 6.07729\, +0. i \\
\end{array}
\right),
\end{footnotesize}\nn\\
&&\begin{footnotesize}
a=\left(
\begin{array}{ccc}
 2.79442\, +26.2305 i & 14.4042\, -1.54735 i & 1.27623\, +2.29185 i \\
 17.0343\, +21.4595 i & 6.13678\, -4.72818 i & -6.82378+24.1677 i \\
 -6.55401+2.08772 i & -6.0073-29.8274 i & -7.59207-24.0165 i \\
\end{array}
\right).
\end{footnotesize}
\eea
The reference state $|\Psi\rangle =\frac{1}{\sqrt{3}}\sum_{k=1}^3|k\rangle_A|k\rangle_{\bar A}$. The eigenvalues of $\eta_A$ and $\eta_{\bar A}$ are
\bea
&&\eta_A\to\{32.6819,-7.87014,5.20516\},\nn\\
&&\eta_{\bar A}\to\{26.3549, 17.3493, -11.5099\}.
\eea
Thus they are not positive operators. The eigenvalues of $\mathcal{T}_A$ are
\bea
\lambda_1=0.849706,\quad \lambda_2= 0.075147 - 0.106401 i,\quad \lambda_3= 0.075147 + 0.106401 i.
\eea
The pseudo R\'enyi entropy is positive in this example as shown in Figure \ref{nonpos}
\begin{figure}[h]
	\centering
\captionsetup[subfloat]{farskip=5pt,captionskip=1pt}
\subfloat{
			\includegraphics[width =0.45\linewidth]{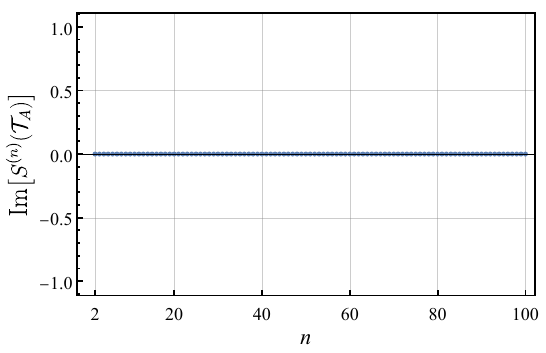}}
\subfloat{
			\includegraphics[width =0.45\linewidth]{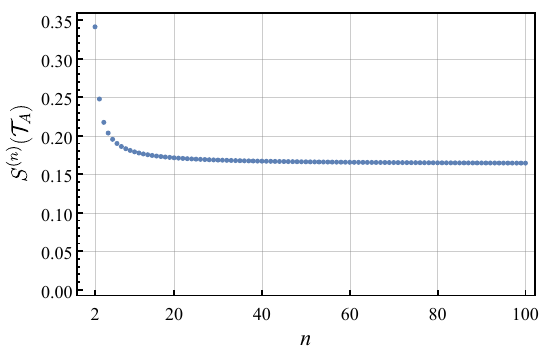}}

\subfloat{
			\includegraphics[width =0.45\linewidth]{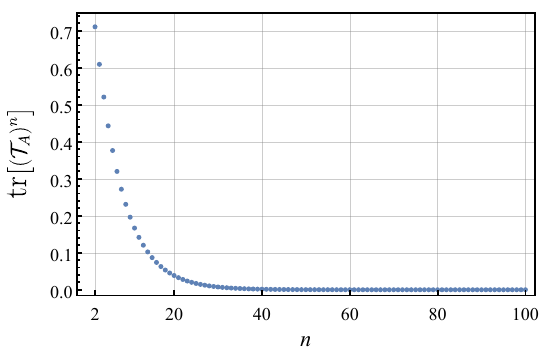}}
\caption{The plot of $S^{(n)}(\mathcal{T}_A)$ and $\tr[(\mathcal{T}_A)^n]$. The upper left plot shows the imaginary part of $S^{(n)}(\mathcal{T}_A)$, which are vanishing. The upper right plot shows $S^{(n)}(\mathcal{T}_A)$. The lower plot shows $\tr[(\mathcal{T}_A)^n]$, which are in the region $(0,1)$. }\label{nonpos}
\end{figure}
\subsection{Example: $S^{(n)}(\mathcal{T}_A)<0$}
In this section we show an example for which $S^{(n)}(\mathcal{T}_A)<0$. Choosing the matrices
\bea
&&\begin{footnotesize}
\eta_A=\left(
\begin{array}{ccc}
 13.9359\, & -17.8554+8.22163 i & 0.740751\, -0.860494 i \\
 -17.8554-8.22163 i & 11.7561\,  & 3.87722\, +0.527719 i \\
 0.740751\, +0.860494 i & 3.87722\, -0.527719 i & 4.32501\,  \\
\end{array}
\right),
\end{footnotesize}\nn\\
&&
\begin{footnotesize}
\eta_{\bar A}=\left(
\begin{array}{ccc}
 2.68826\,  & -2.76297+6.09204 i & -13.4254-5.89942 i \\
 -2.76297-6.09204 i & 23.4288\,  & 2.24652\, -1.6307 i \\
 -13.4254+5.89942 i & 2.24652\, +1.6307 i & 6.07729\, \\
\end{array}
\right),
\end{footnotesize}\nn\\
&&\begin{footnotesize}
a=\left(
\begin{array}{ccc}
 2.79442\, +26.2305 i & 14.4042\, -1.54735 i & 1.27623\, +2.29185 i \\
 17.0343\, +21.4595 i & 6.13678\, -4.72818 i & -6.82378+24.1677 i \\
 -6.55401+2.08772 i & -6.0073-29.8274 i & -7.59207-24.0165 i \\
\end{array}
\right).
\end{footnotesize}
\eea
The reference state $|\Psi\rangle =\frac{1}{\sqrt{3}}\sum_{k=1}^3|k\rangle_A|k\rangle_{\bar A}$. The eigenvalues of $\eta_A$, $\eta_{\bar A}$ and $\mathcal{T}_A$ are given by
\bea
&&\eta_A\to \{85.7965,-45.7377,-0.637431\},\nn\\
&&\eta_{\bar A} \to \{-51.9884,-40.48,28.8633\},\nn\\
&&\mathcal{T}_A\to \{1.37237,-0.368265,-0.00410468\}.
\eea
The pseudo R\'enyi entropy is negative in this example. The result is shown in Figure \ref{negative}.
\begin{figure}[h]
	\centering
\captionsetup[subfloat]{farskip=5pt,captionskip=1pt}
\subfloat{
			\includegraphics[width =0.45\linewidth]{pic2SA.pdf}}
\subfloat{
			\includegraphics[width =0.45\linewidth]{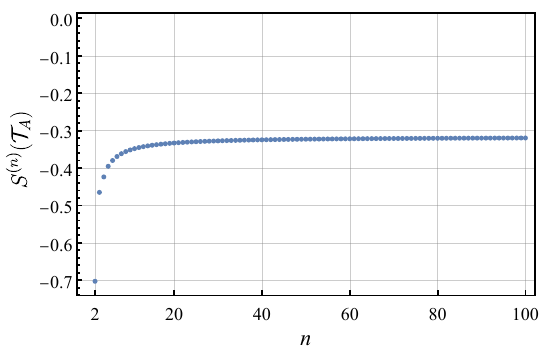}}

\subfloat{
			\includegraphics[width =0.45\linewidth]{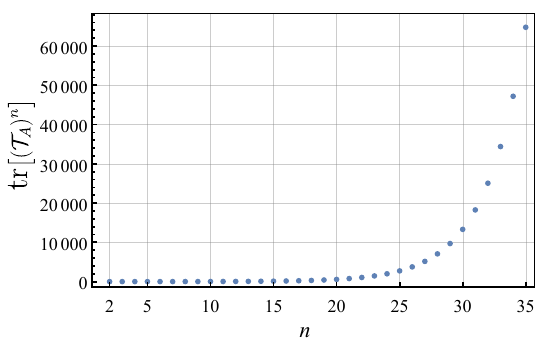}}
\caption{The plot of $S^{(n)}(\mathcal{T}_A)$ and $\tr[(\mathcal{T}_A)^n]$. The upper left plot shows the imaginary part of $S^{(n)}(\mathcal{T}_A)$, which are vanishing. The upper right plot shows $S^{(n)}(\mathcal{T}_A)$, which are negative. The lower plot shows $\tr[(\mathcal{T}_A)^n]$.}\label{negative}
\end{figure}
\subsection{Example with $d=2$}
Let the reference state be $|\Psi\rangle=\frac{1}{\sqrt{2}}(|0\rangle_A |0\rangle_{\bar A}+|1\rangle_A |1\rangle_{\bar A})$. Let the operators $\eta_{A}$ and $\eta_{\bar A}$ be diagonal, $a$ be arbitrary. They are given by
\bea
a=\left(
\begin{array}{cc}
a_{11} & a_{12}\\
a_{21} & a_{22}
\end{array}
\right),
\quad
\eta=\left(
\begin{array}{cc}
\eta_1 & 0\\
0 & \eta_{2}
\end{array}
\right),
\quad
\eta_{\bar A}=\left(
\begin{array}{cc}
\bar \eta_1 & 0\\
0 & \bar \eta_{2}
\end{array}
\right).
\eea
Assume $\eta_A$ and $\eta_{\bar A}$ to be positive, thus $\eta_{1(2)}>0$, $\bar \eta_{1(2)}>0$. One could construct the transition matrix $\mathcal{T}^a$ with these operators. According to Theorem \ref{corol3} we know the eigenvalues of $\mathcal{T}^a_A$ are positive. With some calculations, we have
\bea
\mathcal{T}^a_{ A}=\left(
\begin{array}{cc}
\frac{|a_{11}|^2\eta_1\bar \eta_1+|a_{12}|^2\eta_1\bar \eta_2}{|a_{11}|^2\eta_1\bar \eta_1+|a_{12}|^2\eta_1\bar \eta_2+|a_{21}|^2\bar \eta_1\eta_2+|a_{22}|^2\eta_2\bar \eta_2} & \frac{a_{11} a_{21}^*\eta_2 \bar \eta_1+a_{12} a_{22}^*\eta_2\bar \eta_2}{|a_{11}|^2\eta_1\bar \eta_1+|a_{12}|^2\eta_1\bar \eta_2+|a_{21}|^2\bar \eta_1\eta_2+|a_{22}|^2\eta_2\bar \eta_2}\\
\frac{a_{11}^* a_{21}\eta_1 \bar \eta_1+a_{12}^* a_{22}\eta_1\bar \eta_2}{|a_{11}|^2\eta_1\bar \eta_1+|a_{12}|^2\eta_1\bar \eta_2+|a_{21}|^2\bar \eta_1\eta_2+|a_{22}|^2\eta_2\bar \eta_2} & \frac{|a_{21}|^2\bar \eta_1\eta_2+|a_{22}|^2\eta_2\bar \eta_2}{|a_{11}|^2\eta_1\bar \eta_1+|a_{12}|^2\eta_1\bar \eta_2+|a_{21}|^2\bar \eta_1\eta_2+|a_{22}|^2\eta_2\bar \eta_2}
\end{array}
\right).
\eea
The pseudo R\'enyi entropy of the 2-qubit system is studied in \cite{Nakata:2020luh}. They claim the eigenvalues of $\mathcal{T}^a_A$ are positive if and only if $0\le det[\mathcal{T}^a_A]\le 1/4$. With some calculations, we have
\bea\label{det}
&& \det[\mathcal{T}^a_A]=\frac{|a_{12}a_{21}-a_{11}a_{22}|^2 \eta_1\bar \eta_1\eta_2\bar \eta_2}{(|a_{11}|^2\eta_1\bar \eta_{1}+|a_{12}|^2\eta_1 \bar \eta_2+|a_{21}|^2\bar \eta_1 \eta_2+|a_{22}|^2\eta_2 \bar \eta_2)^2}\nn \\
&&\phantom{det[\mathcal{T}^a_A]}\le \frac{|a_{12}a_{21}-a_{11}a_{22}|^2 \eta_1\bar \eta_1\eta_2\bar \eta_2}{(2|a_{11}||a_{22}|\sqrt{\eta_1\bar \eta_1\eta_2\bar \eta_2}+2|a_{12}||a_{21}|\sqrt{\eta_1\bar \eta_1\eta_2\bar \eta_2})^2}\le \frac{1}{4}.
\eea
The above result can be generalized to arbitrary  positive $\tilde{\eta}_A$ and $\tilde{\eta}_{\bar A}$. Since they are Hermitian operators, there exists unitary operator $U_A$ and $U_{\bar A}$ such that
\bea
\tilde{\eta}_A=U_A \eta_A U_A^\dagger,\quad \tilde{\eta}_{\bar A}=U_{\bar A} \eta_{\bar A} U_{\bar A}^\dagger,
\eea
where $\eta_A$ and $\eta_{\bar A}$ are digonal. The transition matrix $\mathcal{T}^a$ with a given reference state $|\Psi'\rangle$ is given by
\bea
\mathcal{T}^a\propto a |\Psi'\rangle \langle \Psi'| a^\dagger U_A \eta_A U_A^\dagger U_{\bar A} \eta_{\bar A} U_{\bar A}^\dagger.
\eea
Taking partial trace we have
\bea
\mathcal{T}^a_A=tr_{\bar A} \mathcal{T}^a\propto a(tr_{\bar A} U_{\bar A}^\dagger |\Psi'\rangle \langle \Psi' | U_{\bar A}\eta_{\bar A})   a^\dagger U_A \eta_A U_A^\dagger.
\eea
It is always possible to make the operator $tr_{\bar A} U_{\bar A}^\dagger |\Psi'\rangle \langle \Psi' | U_{\bar A}\eta_{\bar A}=tr_{\bar A}  |\Psi\rangle \langle \Psi | \eta_{\bar A}$ by choosing suitable $|\Psi'\rangle$. With this choice one can show $det[\mathcal{T}^a_A]$ is equal to (\ref{det}). Therefore, the transition matrix $\mathcal{T}^a_A$ having positive eigenvalues satisfies that $det[\mathcal{T}^a_A]\le 1/4$, which is consistent with the result in \cite{Nakata:2020luh}.
\section{Details of the example of free scalar with $\theta\ne 0$ }\label{Detailoffreescalar}
In the main text we discuss the transition matrix \eqref{transionunitary}, which can be written as the form
\bea{\label{freescalardetail}}
\mathcal{T}^{\psi|\phi}\propto |\Psi\rangle \langle \Psi| U,
\eea
where $U$ is a unitary operator, it seems the above transition matrix is not like the general form (\ref{Tpseudo}) for the pure pseudo-Hermitian transition matrix. In this section, we will show the transition matrix (\ref{freescalardetail}) can be pseudo-Hermitian for some particular unitary operator $U$ and pure state $|\Psi\rangle$. \\
To satisfy the pseudo Hermitian condition we should require
\bea{}
\eta |\Psi\rangle\langle \Psi| U \eta^{-1} =U^\dagger |\Psi\rangle \langle \Psi|.
\eea
Define the operator $\eta':= U\eta$. This condition is given by
\bea{\label{cond1}}
\eta'|\Psi\rangle =|\Psi\rangle,\quad \langle \Psi| (\eta'^{-1})^\dagger=\langle \Psi|.
\eea
which can be transformed to the operator relation
\bea\label{cond2}
\eta'=\eta'^{-1}+ \alpha P^\Psi_{\bot},
\eea
where $\alpha$ is some constant, $P^\Psi_{\bot}$ satisfies the condition $P^\Psi_{\bot}|\Psi\rangle=0$.  One special case is taking $\alpha=0$. One would have the following relations:
\bea{\label{constraintsetaU}} (\eta')^2=1,\quad  \eta'=(\eta')^{-1}.
\eea
One could check the above two qubits example satisfies the constraints (\ref{cond1}) and (\ref{constraintsetaU}). For the example of free scalar theory with $\theta\ne 0$, one could show the transition matrix is pseudo-Hermitian
by proving the existence of the operator $\eta'$ which satisfies the conditions (\ref{cond1}), (\ref{cond2}), or (\ref{constraintsetaU}).

\section{{Calculation of pseudo R\'enyi entropy by replica method}}\label{app:replica}
 We outline the replica method in QFTs to compute the pseudo R\'enyi entropy in this appendix. In particular, we focus on 2D CFTs as the correlation functions in the replica manifold are easy to obtain by conformal mapping.
%
Let's consider a 2D CFT with  Lagrangian $\mathcal{L}(\phi,\partial\phi) $ dwells on a Euclidean plane $\Sigma_1$ ($ds^2=dwd\bar w$, $(w,\bar{w})=(x+i\tau,x-i\tau)$) and a transition matrix generated by a local operator $\mo(w,\bar w)\equiv e^{\tau H}\mo(x,0)e^{-\tau H}$,
\bea
\mathcal{T}_E^\mo= \frac{\mo(w_1,\bar w_1)|0\rangle \langle 0| \mo^\dagger(w_2,\bar w_2)}{\langle 0| \mo^\dagger(w_2,\bar w_2)\mo(w_1,\bar w_1)|0\rangle},
\eea
where $w_1=x_1-i\tau_1$ and $w_2=x_2+i\tau_2$, $(\tau_1,\tau_2>0)$. The reduced transition matrix of a subsystem $A$, $\mathcal{T}_{E,A}^\mo:=\tr_{\bar A} \mathcal{T}_E^\mo$, can be expressed by path integral with operators inserted at $(w_1,\bar w_1)$ and $(w_2,\bar w_2)$ on the $w$-plane with a cut on $A$
\begin{align}
\langle \phi_{A_-}|\mathcal{T}_{E,A}^\mo|\phi_{A+}\rangle=&\frac{\int^{\phi(x\in A,\tau=0_+)=\phi_{A_+}(x)}_{\phi(x\in A,\tau=0_-)=\phi_{A_-}(x)}[d\phi]\mo^{\dagger}(w_2,\bar w_2)\mo(w_1,\bar w_1)\exp\left\{-\int_{\mathbb{R}^2}\mathcal{L}(\phi,\partial\phi)\right\}}{\int[d\phi]\mo^{\dagger}(w_2,\bar w_2)\mo(w_1,\bar w_1)\exp\left\{-\int_{\mathbb{R}^2}\mathcal{L}(\phi,\partial\phi)\right\}}\nn\\
=&\left(\imineq{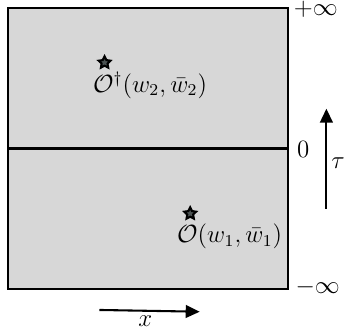}{18}\right)^{-1}\times~~~\imineq{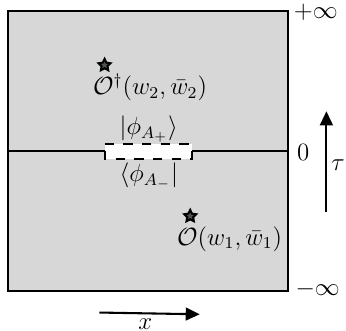}{18}.\label{reduced transition matrix A: graph}
\end{align}
Building on \eqref{reduced transition matrix A: graph}, $\tr[(\mathcal{T}_{E,A}^\mo)^n]$  is  given by a $2n$-point correlation function on a $n$-sheet Riemann surface $\Sigma_n$,
\begin{align}
\tr[(\mathcal{T}_{E,A}^\mo)^n]=&\frac{\mathcal{Z}_n}{\mathcal{Z}_1^n}\cdot\frac{\langle \mo(w_1,\bar w_1) \mo^\dagger(w_2,\bar w_2)...\mo(w_{2n-1},\bar w_{2n-1})\mo^\dagger(w_{2n},\bar w_{2n})\rangle_{\Sigma_n}}{\langle \mo^\dagger(w_2,\bar w_2)\mo(w_1,\bar w_1)\rangle_{\Sigma_1}^n}\nn\\
=&\left(\imineq{TA2.pdf}{18}\right)^{-n}\times~~~\imineq{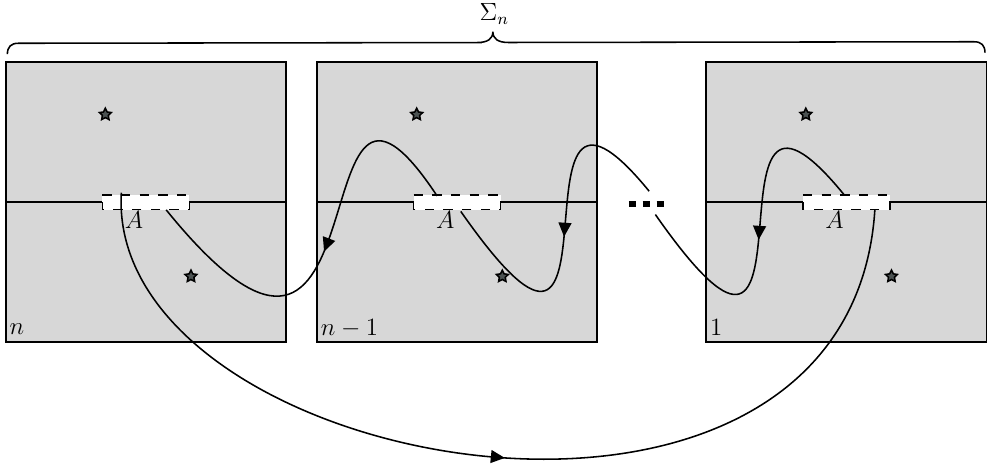}{24}\label{reduced transition matrix An: graph},
\end{align}
where $\mathcal{Z}_1$ and $\mathcal{Z}_n$  are partition functions on $\Sigma_1$ and $\Sigma_n$, respectively, and $\mo(w_{2k-1},\bar w_{2k-1})$ and $\mo^{\dagger}(w_{2k},\bar w_{2k})$  denote the operators inserted at $k$th sheet.
The $n$th pseudo R\'enyi entropy of $\mathcal{T}^{\mo}_{E,A}$  turns out to be
\bea
S^{(n)}(\mathcal{T}_{E,A}^\mo)=S^{(n)}_{A;vac}+\Delta S^{(n)}(\mathcal{T}_{E,A}^\mo),
\eea
where $S^{(n)}_{A;vac}\equiv\frac{1}{1-n}\log\frac{\mathcal{Z}_n}{\mathcal{Z}_1^n}$ is the $n$th R\'enyi entropy of $A$ when the total system is in  the vacuum, and $\Delta S^{(n)}(\mathcal{T}_{E,A}^\mo)$ we refer to as the excess of $n$th pseudo R\'enyi entropy of $A$,
\bea
\Delta S^{(n)}(\mathcal{T}_{E,A}^\mo)=\frac{1}{1-n}\log \frac{\langle \mo(w_1,\bar w_1) \mo^\dagger(w_2,\bar w_2)...\mo(w_{2n-1},\bar w_{2n-1})\mo^\dagger(w_{2n},\bar w_{2n})\rangle_{\Sigma_n}}{\langle\mo(w_1,\bar w_1) \mo^\dagger(w_2,\bar w_2)\rangle_{\Sigma_1}^n}.
\eea
For our purposes, we only focus on the $2$nd pseudo R\'enyi entropy,
\bea
\Delta S^{(2)}(\mathcal{T}_{E,A}^\mo)=-\log \frac{\langle \mo(w_1,\bar w_1) \mo^\dagger(w_2,\bar w_2)\mo(w_3,\bar w_3)\mo^\dagger(w_4,\bar w_4)\rangle_{\Sigma_2}}{\langle \mo (w_1,\bar w_1)\mo^\dagger(w_2,\bar w_2)\rangle_{\Sigma_1}^2}\label{appf:deltaSA2}.
\eea
Meanwhile,  $\mo$ is assumed to be a primary with chiral and anti-chiral conformal dimension $\Delta_\mo$.
 By conformal symmetry, the 2- and 4-point function of $\mo$ on $\Sigma_1$  can be expressed as
\begin{align}
\langle\mo(z_1,\bar{z}_1)\mo^{\dagger}(z_2,\bar{z}_2)\rangle_{\Sigma_1}&=\frac{c_{12}}{|z_{12}|^{4\Delta_\mo}},\label{2pt function on sigma1}\\
\langle\mo(z_1,\bar{z}_1)\mo^{\dagger}(z_2,\bar{z}_2)\mo(z_3,\bar{z}_3)\mo^{\dagger}(z_4,\bar{z}_4)\rangle_{\Sigma_1}&=|z_{13}z_{24}|^{-4\Delta_\mo}G(\eta,\bar{\eta}),\label{4ptandG}
\end{align}
respectively, where $\eta:=\frac{z_{12}z_{34}}{z_{13}z_{24}}$ and $\bar \eta:=\frac{\bar z_{12}\bar z_{34}}{\bar z_{13}\bar z_{24}}$ are the cross ratios. Since there are conformal mappings
\bea\label{coordinatetrans}
z=
\begin{cases}
 w^{1/n},&\quad A=[0,\infty),\nn\\
 \left( \frac{w+L}{w-L} \right)^{1/n},&\quad A=[-L,L],
\end{cases}
\eea
from $\Sigma_n$ to $\Sigma_1$,   the 4-point function on $\Sigma_2$ is obtained by applying the above conformal mappings with $n=2$
\begin{equation}
\langle\mo(w_1,\bar{w}_1)\mo^{\dagger}(w_2,\bar{w}_2)\mo(w_3,\bar{w}_3)\mo^{\dagger}(w_4,\bar{w}_4)\rangle_{\Sigma_2}=
\begin{cases}
\big|\frac{64L^2z_1^2z_2^2}{(z_1^2-1)^2(z_2^2-1)^2}\big|^{-4\Delta_\mo}G(\eta,\bar{\eta}),~&A=[-L,L],\\
\\
\big|16z_1^2z_2^2\big|^{-4\Delta_\mo}G(\eta,\bar{\eta}),~&A=[0,+\infty)\label{4pt on sigma2}.
\end{cases}
\end{equation}
Substituting \eqref{2pt function on sigma1} and \eqref{4pt on sigma2} into \eqref{appf:deltaSA2} and after some algebra,
we obtain
\bea\label{S2formula}
\Delta S^{(2)}(\mathcal{T}_{E,A}^\mo)=\log \frac{c_{12}^2}{|\eta(1-\eta)|^{4\Delta_\mo} G(\eta,\bar \eta)},
\eea
which only depends on the cross ratios $\eta$ and $\bar \eta$. The 2nd pseudo R\'enyi entropy with regard to the real-time dependent transition matrix can be obtained by applying the analytic continuation to $\tau_1$ and $\tau_2$ in the above result. When $\tau_1\to\e+ it$ and $\tau_2\to\epsilon -i t$, we meet the case studied in \cite{Guo:2022sfl}. As we mentioned in the previous section, we would like to focus on the case of $\mathcal{T}^{\mathcal{O}}=\frac{\mo(x,t)|0\rangle\langle0|\mo(-x,-t)}{\langle0|\mo(-x,-t)\mo(x,t)|0\rangle}$. Thus we have the analytic continuation $\tau_1=\tau_2\to\epsilon -i t$. An infinitesimally small regularization parameter $\e$ is introduced to suppress the high energy modes \cite{Calabrese:2005in}.
\bibliographystyle{JHEP}
\bibliography{PE-bib}{}
\end{document}